\documentclass[conference, 11pt, onecolumn]{ieeeconf}

\IEEEoverridecommandlockouts
\usepackage[usenames]{color}
\usepackage{enumerate}
\usepackage{url}
\usepackage{subfigure}
\usepackage{amsfonts,mathrsfs}
\usepackage{amssymb,amsmath}
\usepackage{verbatim}
\usepackage{acronym}
\usepackage{mathtools}
\usepackage{cite}
\usepackage{graphicx}
\usepackage{algorithm}
\usepackage[noend]{algpseudocode}

\def\fskip#1{}

\newtheorem{theorem}{Theorem}

\newtheorem{assumption}{Assumption}

\newtheorem{corollary}{Corollary}

\newtheorem{definition}{Definition}

\newtheorem{lemma}{Lemma}

\newtheorem{remark}{Remark}

\renewcommand{\theenumi}{\arabic{enumi}}

\def\1{{\bf 1}}

\newcommand{\remove}[1]{}

\def\argmin{\mathop{\rm argmin}}
\def\argmax{\mathop{\rm argmax}}

\begin{document}
\title{Maximizing Convergence Time in Network Averaging Dynamics Subject to Edge Removal}
\vspace{-0.5cm}
\author{\authorblockN{S. Rasoul Etesami*\vspace{-0.5cm}}
\thanks{Department of Industrial and Systems Engineering and Coordinated Science Laboratory, University of Illinois at Urbana-Champaign,  Urbana, IL 61801, (Email: etesami1@illinois.edu).  This work is supported by the NSF CAREER Award under Grant No. EPCN-1944403. }
}
\maketitle
\begin{abstract}
We consider the consensus interdiction problem (CIP), in which the goal is to maximize the convergence time of consensus averaging dynamics subject to removing a limited number of network edges. We first show that CIP can be cast as an effective resistance interdiction problem (ERIP), in which the goal is to remove a limited number of network edges to maximize the effective resistance between a source node and a sink node. We show that ERIP is strongly NP-hard, even for bipartite graphs of diameter three with fixed source/sink edges, and establish the same hardness result for the CIP. We then show that both ERIP and CIP cannot be approximated up to a (nearly) polynomial factor assuming exponential time hypothesis. Subsequently, we devise a polynomial-time $mn$-approximation algorithm for the ERIP that only depends on the number of nodes $n$ and the number of edges $m$,  but is independent of the size of edge resistances.  Finally,  using a quadratic program formulation for the CIP, we devise an iterative approximation algorithm to find a first-order stationary solution for the CIP and evaluate its good performance through numerical experiments.
\end{abstract}

\section{Introduction}
Robustifying network operations in the presence of non-ideal physical nature (such as switching topologies and transmission delays  \cite{olfati2004consensus,nedic2010convergence}; noisy links \cite{xiao2007distributed}; quantization \cite{etesami2016convergence,bacsar2016convergence}; random link failures \cite{kar2009distributed}), as well as securing them against misbehavior of nodes and malicious adversarial attacks have been important areas of research in network security \cite{alpcan2010network}. Most research in this arena has been focused on applications such as consensus formation, or distributed averaging, and selected behavior of the adversary that leads to disruption of the underlying operation \cite{khanafer2015robust,pasqualetti2011consensus}.  In this paper, we shall follow the same path and analyze the impact of network interdiction on the convergence time of the consensus averaging dynamics.  

One of the main motivations for considering averaging dynamics as the underlying operation is that such dynamics are building blocks for many complex network problems. For instance, distributed optimization over networks \cite{nedic2018multi,nedic2010constrained}, coverage control \cite{schwager2008consensus}, formation control \cite{olfati2006flocking}, distributed Kalman filtering \cite{carli2008distributed},  and load balancing \cite{nedic2009distributed} are some examples of control and coordination problems with proposed solutions that rely crucially on distributed averaging or consensus dynamics. Therefore, a better
understanding of the network averaging dynamics in the presence of an adversary contributes toward more efficient design methods for robustifying the network against adversarial attacks.  As one of the immediate applications of this work, one can consider the security of power networks \cite{teixeira2010networked}, in which an adversary aims to sabotage the system's operation by interdicting the power lines (e.g., to imbalance the load/voltage at different nodes). As another application, one can consider the control of network diffusion dynamics in which the goal is to slow down the speed of diffusion dynamics by removing a limited number of network edges. For instance, blocking rumor or misinformation over social networks by removing an edge set \cite{yan2019rumor,jia2020blocking,zareie2021minimizing} or containing epidemics over socio-biological networks \cite{pare2018analysis} by restricting social interactions are two typical examples of control of diffusion dynamics over networks.

In this paper, our goal is to analyze the influence of network interdiction on operations of the consensus averaging dynamics from a computational perspective. In particular, we consider an adversary whose impact on the network could be breaking some selected links to prevent or slow down communication among corresponding agents, hence increasing the convergence time to consensus.  Without the presence of an adversary, distributed averaging involves the computation of the average of the initial values stored at each node of a graph in a distributed way, where, in an evolving process, each node (agent) updates its value as a linear combination of the values of its neighbors. In discrete-time, the process can be expressed in terms of a linear difference equation $x(t+1) = P x(t), x(0) = x_0$, where $x(t)=(x_1(t),\ldots,x_n(t))'$ is the column vector stacking all the nodes' values $x_i(t), i=1,2,\ldots,n$ at time $t$ (in a graph with $n$ nodes), $x_0$ is the vector of initial values, and $P=(p_{ij})$ is a symmetric irreducible stochastic matrix, with $p_{ij} = 0$ if and only if $\{i, j\}$ is not an edge of the graph. This process converges to the true average (and therefore all nodes reach a consensus on this true average), that is, $\lim_{t\to\infty} x_i(t) = \frac{\sum_{j=1}^n x_j(0)}{n}, \forall i$. The rate of the convergence depends on the network topology and the weights $p_{ij}$'s. Now, if we bring into the picture an adversary who wants to disrupt that process, we can endow the adversary with the capability to break a limited number of $\ell$ links to prevent convergence to the true average or slow down the convergence. To measure the convergence speed, one can use a well-studied objective function \cite{khanafer2015robust,lovisari2013resistance} $\frac{1}{2}\sum_{t=0}^{\infty} k(t) \|x(t) - \bar{x}\|^2$, where $k(\cdot)$ is a positive weighting function, and the adversary wants to maximize it by deciding what $\ell$ links to break at each time instance $t$. As we show in this work, finding the adversary's optimal strategy in this setting is strongly NP-hard and cannot be approximated within a nearly polynomial factor.

In this work, we introduce two new concepts: i) \emph{consensus interdiction problem} (CIP), in which the goal is to maximize the convergence time of consensus dynamics subject to removing a limited number of network edges; and ii) \emph{effective resistance interdiction problem} (ERIP), in which the goal is to remove a limited number of network edges to maximize the effective resistance between a source and a sink node.  We show a close connection between these two problems and leverage that connection to establish our complexity and inapproximability results. In particular, we develop approximation algorithms for both ERIP and CIP. To the best of our knowledge, this work is the first to study CIP and ERIP within the context of averaging dynamics.

\subsection{Related Work}

This work is related to the modeling framework described in \cite{khanafer2015robust}, which adopted the setting of distributed averaging as network operation.  It was shown in \cite{khanafer2015robust} that under a more elaborate setting, where there is a network designer/adversary who can repair/break a limited number of edges, the game between the adversary and the network designer admits a saddle-point strategy.\footnote{It was erroneously shown in \cite{khanafer2015robust} that under this more elaborate setting, the adversary's optimal strategy could be found in polynomial time. We give a counterexample to this claim in Appendix II.} Our work is also related to the literature on network interdiction problems \cite{smith2020survey}, in which, broadly speaking, the goal is to know how sensitive a particular property of a network is with respect to changes in the graph structure (e.g., edge removal). We mention here matching interdiction \cite{zenklusen2010matching}, edge connectivity interdiction \cite{zenklusen2014connectivity}, network flow interdiction \cite{wood1993deterministic,chestnut2017hardness}, minimum spanning tree interdiction \cite{zenklusen20151}, and shortest path interdiction \cite{bar1995complexity} as some of the well-studied network interdiction problems. Expanding upon the past literature, in this paper, we introduce the ERIP and CIP. To establish the NP-hardness of the ERIP and CIP, even when we are allowed to break a subset of edges, we use a reduction from the max-clique problem \cite{karp1972reducibility}. It is worth noting that \cite{wood1993deterministic} also uses a reduction from the max-clique problem to establish the complexity of the network flow interdiction problem, in which the goal is to minimize the maximum flow by interdicting a limited number of network edges. However, establishing complexity results for the ERIP and CIP is more challenging because, unlike the network flow problem that admits a linear program formulation, the objective functions in ERIP and CIP have more general convex formulations and do not necessarily admit integral flows. Moreover, the existence of edge capacities in the network flow interdiction problem simplifies the complexity analysis as it allows one to control the flow directly via capacity adjustment. We also use another reduction from densest $k$-subgraph (D$k$S) \cite{manurangsi2017almost,chestnut2017hardness} to establish inapproximability results for both ERIP and CIP.

This paper is also related to the path interdiction problem \cite{bar1995complexity,ball1989finding,khachiyan2008short,schieber1995complexity}, in which the goal is to cut a limited number of edges in a graph to maximize the length of the shortest path between two terminal nodes. However, the effective resistance between two nodes is a complicated function of all the paths between those nodes. Therefore, it is not clear how the complexity results from the path interdiction problem can be carried over to the ERIP or CIP.  Moreover, since the objective functions in ERIP and CIP have quadratic forms, it is tempting to view these problems as a special case of sesquilinear programming \cite{toker1996complexity}. Unfortunately,  this is not true because, in sesquilinear programming, the objective function is of the form $z'Qz$, where $Q$ is a fixed positive definite matrix, and $z$ is the optimizing variable. In contrast, in the ERIP and CIP, $z$ is a fixed vector and the optimizing variable impacts entries of the matrix $Q$ in a certain pattern.  Finally, we note that the minimization version of the ERIP has been recently studied in \cite{chan2019network}, where it was shown that minimizing the effective resistance by buying at most a limited number of edges of an underlying network is NP-hard. However, minimization and maximization of the effective resistance exhibit completely different behaviors and the results from one side cannot be readily applied to the other side.  In fact, effective resistance is an important measure with a wide range of applications such as network robustness \cite{wang2014improving}, performance analysis of consensus algorithms \cite{lovisari2013resistance,etesami2016convergence}, spanning tree enumeration \cite{li2019maximizing}, commute time and mixing times \cite{levin2017markov}, graph eigenvalue optimization \cite{boyd2006convex}, epidemics over networks \cite{etesamipotential}, and power dissipation \cite{ghosh2008minimizing}. Therefore, we believe that our complexity results may directly impact many other applications.

\subsection{Organization and Notations}
The paper is organized as follows. In Section \ref{sec:introduction}, we formally introduce the CIP and establish its connection with the ERIP. In Section \ref{sec:complexity}, we show NP-hardness of the ERIP and CIP, even under a restrictive setting with fixed source/sink edges. In Section \ref{sec:hardness-approx}, we establish the hardness of approximation for the ERIP and CIP.  We devise approximation algorithms for both ERIP and CIP in Section \ref{sec:quadratic}. We provide some numerical results in Section \ref{sec:numerical} and conclude the paper in Section \ref{sec:conclusion}. Auxiliary lemmas are given in Appendix I, and a correction to the past literature is given in Appendix II.

\smallskip
\noindent
{\bf Notations:} We let $\boldsymbol{{\rm e}}_i$ be the $i$th Euclidean basis and $\boldsymbol{1}$ be a column vector of all ones. We let $J=\boldsymbol{1}\boldsymbol{1}'$ be a square matrix with all entries being $1$. A Laplacian matrix is a square matrix with nonnegative diagonal entries such that the sum of the entries in each row equals zero. We use $L^{+}$ to denote the pseudoinverse of a Laplacian matrix $L$, which is a square matrix satisfying $L^{+}L=I-\frac{J}{n}$. Given an undirected graph $\mathcal{G}=(V, E)$ with vertex set $V$ and edge set $E$, we denote the subgraph induced on $S\subseteq V$ by $\mathcal{G}[S]$. Moreover, we denote the set of edges with both endpoint in $S$ by $E[S]$, and the set of edges with only one endpoint in $S$ by $\delta(S)$. Similarly, for $S,T\subset V, S\cap T=\emptyset$, we denote the set of edges between $S$ and $T$ by $E[S,T]$. Given a network with conductance matrix $P=(p_{ij})$, i.e., an edge $\{i,j\}$ has conductance $p_{ij}$ (or resistance $r_{ij}=\frac{1}{p_{ij}}$), we denote the effective resistance between the fixed nodes $s$ and $t$ by $\mbox{R}_{\rm eff}(P)$.  In other words,  $\mbox{R}_{\rm eff}(P)$ denotes the induced voltage between $s$ and $t$ when a unit electric current is inserted into node $s$ and extracted from node $t$.


\section{Problem Formulation and Preliminary Results}\label{sec:introduction}
Let us consider the discrete-time consensus dynamics:
\begin{align}\label{eq:consensus-dynamics}
x(t+1)&=Px(t) \ \ t=0,1,2,\ldots, \cr 
x(0)&=x_0,
\end{align}
where $x_0$ is the $n$-dimensional vector of initial values (in a graph with $n$ nodes), and $P=(p_{ij})$ is an irreducible and symmetric stochastic matrix,\footnote{A nonnegative matrix is called stochastic if the sum of entries in each row equals $1$. } with $p_{ij} = 0$ if and only if $\{i, j\}$ is not an edge of the graph. It is known that the above dynamics will converge to the consensus vector given by $\bar{x}=\frac{1}{n}Jx_0$, where $J$ is the $n\times n$ matrix of all ones. A well-known performance index for measuring the speed of convergence of the consensus dynamics \eqref{eq:consensus-dynamics} to their equilibrium point $\bar{x}$ is given by the aggregate deviation of the iterates from the consensus point \cite{khanafer2015robust,lovisari2013resistance}, i.e.,   
\begin{align}\label{eq:index}
\mathcal{J}(P, x_0)=\sum_{t=0}^{\infty}\|x(t)-\bar{x}\|^2.
\end{align}  

Given an integer budget $\ell\in \mathbb{Z}_+$, and a subset of at most $\ell$ edges $\hat{E}\subseteq \big\{\{i,j\}: p_{ij}>0, i\neq j\big\}, |\hat{E}|\leq \ell$, let us define $P\setminus\hat{E}$ to be the symmetric stochastic matrix that is obtained from $P$ by removing the edges in $\hat{E}$, where removing an edge shifts the weight of that edge to its endpoints.  The reason for such a weight shift is that in distributed averaging dynamics, the sum of weights emanating from each node must always equal $1$. Thus, when an edge is broken, the endpoints of that edge will no longer observe each other, hence returning that weight into themselves as self-loops. More precisely, if $\{i,j\}\in \hat{E}$, then the $ij$-th entry of the matrix $P\!\setminus\!\hat{E}$ is given by 
\begin{align}\nonumber
&\big(P\!\setminus\!\hat{E}\big)_{ij}=\big(P\!\setminus\!\hat{E}\big)_{ji}=0,\cr 
&\big(P\!\setminus\!\hat{E}\big)_{ii}=p_{ii}+p_{ij},\cr 
&\big(P\!\setminus\!\hat{E}\big)_{jj}=p_{jj}+p_{ij}.
\end{align}
For any $\{i,j\}\notin \hat{E}$, we set $\big(P\setminus\hat{E}\big)_{ij}=p_{ij}$. It is easy to see that the modified matrix $P\setminus\hat{E}$ is also symmetric and stochastic but without the edges in $\hat{E}$. 

\begin{assumption}\label{ass:l-connectivity}
We assume that the budget $\ell$ is smaller than the network edge-connectivity so that by removing $\hat{E}$, the network
remains connected. This assumption is needed for the stability of the consensus dynamics such that \eqref{eq:index} remains finite.
\end{assumption}

In this work, we want to know that given a general symmetric stochastic matrix $P$, and an initial vector $x_0$, what is the optimal set of $\ell$ edges whose deletion from $P$ maximizes the objective function \eqref{eq:index}. In other words, we want to maximize the convergence time of the consensus dynamics \eqref{eq:consensus-dynamics} by removing at most $\ell$ edges from the underlying network. That brings us to the following optimization problem:
\noindent
\begin{algorithm}
{\bf Consensus Interdiction Problem (CIP)}: Given an irreducible and symmetric stochastic matrix $P$, an arbitrary initial vector $x_0$, and an integer budget $\ell\in \mathbb{Z}_+$, find an edge cut $\hat{E}$ of at most $\ell$ edges that solves the following optimization problem: 
\begin{align}\label{eq:optimization-problem}
&\max_{|\hat{E}|\leq \ell} \mathcal{J}\big(P\!\setminus\!\hat{E},x_0\big)=\sum_{t=0}^{\infty}\|\tilde{x}(t)-\bar{x}\|^2\cr 
&\mbox{s.t.}\ \ \ \tilde{x}(t+1)=\big(P\!\setminus\!\hat{E}\big)\tilde{x}(t), \ \tilde{x}(0)=x_0.
\end{align}
\end{algorithm}

\begin{lemma}\label{lemm:eff-cast}
Given any initial vector $x_0\in \mathbb{R}^n$, the CIP \eqref{eq:optimization-problem} is equivalent to solving the optimization problem $\max_{|\hat{E}|\leq \ell} \ x'_0\big(I-(P\!\setminus\!\hat{E})^2+\frac{J}{n}\big)^{-1}x_0$. In particular, for the specific choice of initial vector $x_0=\boldsymbol{{\rm e}}_s-\boldsymbol{{\rm e}}_t$, the CIP \eqref{eq:optimization-problem} is equivalent to solving the optimization problem $\max_{|\hat{E}|\leq \ell} \mbox{{\rm R}}_{\rm eff}\big((P\!\setminus\!\hat{E})^2\big)$.
\end{lemma}
\begin{proof}
For simplicity, let us define $\tilde{P}=P\setminus\hat{E}$. Using the definition of the consensus vector $\bar{x}=\frac{1}{n}Jx_0$ in the objective function \eqref{eq:optimization-problem}, we have
\begin{align}\nonumber
\mathcal{J}(\tilde{P},x_0)=\sum_{t=0}^{\infty}\|\tilde{x}(t)-\bar{x}\|^2=\sum_{t=0}^{\infty}\|\tilde{P}^tx_0-\frac{J}{n}x_0\|^2=\sum_{t=0}^{\infty}\|(\tilde{P}^t-\frac{J}{n})x_0\|^2.
\end{align}
Since $\tilde{P}$ is a symmetric stochastic matrix, for any $t\in \mathbb{Z}_+$, we have $\tilde{P}^tJ=J\tilde{P}^t=J$. Thus, using $(\frac{J}{n})^t=\frac{J}{n}, \forall t$, we have $\tilde{P}^t-\frac{J}{n}=(\tilde{P}-\frac{J}{n})^t, \forall t$, and we can write
\begin{align}\nonumber
\mathcal{J}(\tilde{P},x_0)&=\sum_{t=0}^{\infty}\|(\tilde{P}^t-\frac{J}{n})x_0\|^2=\sum_{t=0}^{\infty}\|(\tilde{P}-\frac{J}{n})^tx_0\|^2\cr 
&=
x'_0\big(\sum_{t=0}^{\infty}(\tilde{P}-\frac{J}{n})^{2t}\big)x_0=x'_0\big(\sum_{t=0}^{\infty}(\tilde{P}^2-\frac{J}{n})^{t}\big)x_0\cr 
&=x'_0\big(I-\tilde{P}^2+\frac{J}{n}\big)^{-1}x_0.
\end{align}
Therefore, given $P, x_0,\ell$, the CIP \eqref{eq:optimization-problem} can be written as
\begin{align}\label{eq:P-J}
\max_{|\hat{E}|\leq \ell} \ x'_0\big(I-\tilde{P}^2+\frac{J}{n}\big)^{-1}x_0,
\end{align}
which completes the first part of the proof. To show the second part, let us denote the Laplacian matrix associated with the conductance matrix $\tilde{P}^2$ by $L=I-\tilde{P}^2$. Then, using \cite[Eq. (7)]{ghosh2008minimizing}, the pseudoinverse of the Laplacian matrix $L$ is given by $L^{+}=(L+\frac{J}{n})^{-1}-\frac{J}{n}$. On the other hand, it is known \cite[Eq. (9)]{ghosh2008minimizing} that the effective resistance between nodes $s$ and $t$ in an electric network with conductance matrix $\tilde{P}^2$ (i.e., edge $\{i,j\}$ has conductance $(\tilde{P}^2)_{ij}$) is given by $\mbox{R}_{\rm eff}(\tilde{P}^2)=(\boldsymbol{{\rm e}}_s-\boldsymbol{{\rm e}}_t)'L^+(\boldsymbol{{\rm e}}_s-\boldsymbol{{\rm e}}_t)$. Therefore, we can write
\begin{align}\nonumber
\mbox{R}_{\rm eff}(\tilde{P}^2)&=(\boldsymbol{{\rm e}}_s-\boldsymbol{{\rm e}}_t)'L^+(\boldsymbol{{\rm e}}_s-\boldsymbol{{\rm e}}_t)\cr 
&=(\boldsymbol{{\rm e}}_s-\boldsymbol{{\rm e}}_t)'\Big((L+\frac{J}{n})^{-1}-\frac{J}{n}\Big)(\boldsymbol{{\rm e}}_s-\boldsymbol{{\rm e}}_t)\cr 
&=(\boldsymbol{{\rm e}}_s-\boldsymbol{{\rm e}}_t)'(L+\frac{J}{n})^{-1}(\boldsymbol{{\rm e}}_s-\boldsymbol{{\rm e}}_t)\cr 
&=(\boldsymbol{{\rm e}}_s-\boldsymbol{{\rm e}}_t)'(I-\tilde{P}^2+\frac{J}{n})^{-1}(\boldsymbol{{\rm e}}_s-\boldsymbol{{\rm e}}_t),
\end{align}
where the third equality holds because $J(\boldsymbol{{\rm e}}_s-\boldsymbol{{\rm e}}_t)=0$. Thus, if we choose the initial vector in \eqref{eq:P-J} to be $x_0=\boldsymbol{{\rm e}}_s-\boldsymbol{{\rm e}}_t$, the CIP reduces to solving $\max_{|\hat{E}|\leq \ell} \ \mbox{R}_{\rm eff}\big((P\setminus\hat{E})^2\big)$.
\end{proof}

As is shown in Lemma \ref{lemm:eff-cast}, the CIP is closely related to the effective resistance of the \emph{squared} conductance matrix $(P\setminus\hat{E})^2$. For that reason, we introduce the effective resistance interdiction problem (ERIP), which is similar to the CIP except that the interdicted conductance matrix is given by $P\setminus\hat{E}$ (rather than $(P\setminus\hat{E})^2$). In Section \ref{sec:scaling}, we will use a scaling argument to connect these two problem together.

\begin{algorithm}
\noindent
{\bf Effective Resistance Interdiction Problem (ERIP)}: Given a resistance network with conductance matrix $P$, two identified terminals $s$ and $t$, and an integer budget $\ell\in \mathbb{Z}_+$, find an edge cut $\hat{E}$ of at most $\ell$ edges that solves the optimization:
\begin{align}\label{eq:ERIP}
\max_{|\hat{E}|\leq \ell} \mbox{R}_{\rm eff}(P\!\setminus\!\hat{E}).
\end{align}
\end{algorithm}

\begin{remark}\label{rem:ERIP}
Using the same argument as in deriving \eqref{eq:P-J}, one can see that solving the ERIP is equivalent to solving the following optimization problem
\begin{align}\label{eq:E-R-J}
\max_{|\hat{E}|\leq \ell} \ (\boldsymbol{\mbox{e}}_s-\boldsymbol{\mbox{e}}_t)'\Big(I-(P\!\setminus\!\hat{E})+\frac{J}{n}\Big)^{-1}(\boldsymbol{\mbox{e}}_s-\boldsymbol{\mbox{e}}_t).
\end{align}
\end{remark}

\section{Complexity of the ERIP and CIP}\label{sec:complexity}

In this section, we first establish the strong NP-hardness of solving the ERIP and postpone its extension to the CIP to subsection \ref{sec:scaling}. Here, the strong NP-hardness means that solving the ERIP remains NP-hard even if all the entries of the conductance matrix $ P $ are bounded by a polynomial function of the number of nodes $n$. For simplicity of presentation, we work directly with edge resistances $r_{ij}=\frac{1}{p_{ij}}$ rather than edge conductances $p_{ij}$. 

\begin{definition}
Given a \emph{directed} network $\vec{\mathcal{G}}=(V,\vec{E})$ with directed edge set $\vec{E}$ and two nodes $s,t\in V$, a \emph{directed flow} from $s$ to $t$ is a nonnegative function $f:\vec{E}\to \mathbb{R}_+$ that satisfies flow conservation constraints, i.e., $\sum_{i: (i,j)\in \vec{E}}f_{ij}=\sum_{i: (j,i)\in \vec{E}}f_{ji}$, $\forall j\in V\setminus\{s,t\}$. The strength of the directed flow $f$ is defined to be $\|f\|=\sum_{i: (s,i)\in \vec{E}}f_{si}$. A \emph{unit directed flow} is a flow with strength $1$.    
\end{definition}

\begin{remark}
Given a unit directed flow $f$ from $s$ to $t$ in $\vec{\mathcal{G}}=(V,\vec{E})$, we often abuse the notation and define its (undirected) flow function $f:E\to \mathbb{R}_+$ on the undirected network $\mathcal{G}=(V,E)$ by simply the value of the directed flow on the edges regardless of the edge orientations. In other words, for each undirected edge $e=\{i,j\}\in E$ we let $f_e=f_{ij}$ or $f_e=f_{ji}$ depending on whether $(i,j)\in \vec{E}$ or $(j,i)\in \vec{E}$.
\end{remark}

Next, we state the following well-known lemma, which allows us to upper-bound the effective resistance using energy dissipation of unit flows. 

\begin{lemma}\label{lemm:thomson}(Thomson's Principle \cite[Theorem 9.10]{levin2017markov}) The effective resistance between $s$ and $t$ is the minimum energy dissipation over the network $\mathcal{G}=(V,E)$ by sending a unit of flow from $s$ to $t$, i.e.,  
\begin{align}\nonumber
\mbox{{\rm R}}_{\rm eff}(P)=\min\Big\{\sum_{e\in E}r_e f^2_e: f \ \mbox{is a unit flow from $s$ to $t$}\Big\},
\end{align}
where $r_e=\frac{1}{p_{ij}}$ denotes resistance of the edge $e=\{i,j\}\in E$. In particular, the minimum energy is achieved for the unit \emph{electrical} flow that also satisfies Ohm's laws. 
\end{lemma}

A simple corollary of Thomson's principle is the following known result
\begin{corollary}\label{cor:ray}(Rayleigh's Monotonicity Law \cite[Theorem 9.12]{levin2017markov}) 
The $s-t$ effective resistance cannot increase if the resistance of an edge is decreased. In particular, adding an edge does not increase the effective resistance. 
\end{corollary}

Let us consider the decision version of the effective resistance interdiction problem that, with abuse of notation, we denote it again by ERIP. The decision problem is given by a resistance network $\mathcal{G}=(V', E', \{r_e\}_{e\in E'})$, fixed terminals $s, t\in V$, and two positive numbers $R_0\in \mathbb{R}_+, \ell\in \mathbb{Z}_+$. The goal is to decide whether there exists a subset $\hat{E}\subseteq E'$ of at most $\ell$ edges whose removal from $\mathcal{G}'$ increases the effective resistance between $s$ and $t$ to a value higher than $R_0$, i.e., $\mbox{R}_{{\rm eff}}(\mathcal{G}'\setminus \hat{E})\ge R_0$. 

Next, we describe the network construction in our complexity reduction. Given an arbitrary undirected graph $\mathcal{G}=(V, E), r\in \mathbb{Z}_+$ with $|V|=n$ nodes and $|E|=m$ edges, we construct a resistance network $\mathcal{G}'=(V', E',\{r_e\}_{e\in E'})$ with $|V'|=n+m+2$ nodes and $|E'|=3m+n$ edges as follows. For each undirected edge $\{i,j\}\in E$, we put one vertex $v_{ij}=v_{ji}$ on the left side of a bipartite graph, and for each node $i\in V$, we put one vertex $v_i$ on the right side of that bipartite graph. We denote the vertices in the left and right side of that bipartite graph by $V_L=\{v_{ij}: \{i,j\}\in E\}$ and $V_R=\{v_i: i\in V\}$, respectively. We connect $v_{ij}$ to exactly two nodes $v_i$ and $v_j$, and define $E_1$ to be the set of all such edges, i.e., $E_1=\{\{v_{ij},v_i\}: \{i,j\}\in E, i\in V\}$. Moreover, we add two additional nodes $s$ and $t$, where $s$ is connected to all the vertices in $V_L$, and $t$ is connected to all the vertices in $V_R$, and we define $E_L=\{\{s,v_{ij}\}: \{i,j\}\in E\}$ and $E_R=\{\{t,v_i\}: i\in V\}$. An example of the above construction is given in Figure \ref{fig:reduction}.

\begin{figure}[t]
\vspace{-2.5cm}
\begin{center}
\includegraphics[totalheight=.25\textheight,
width=.35\textwidth,viewport=100 0 800 700]{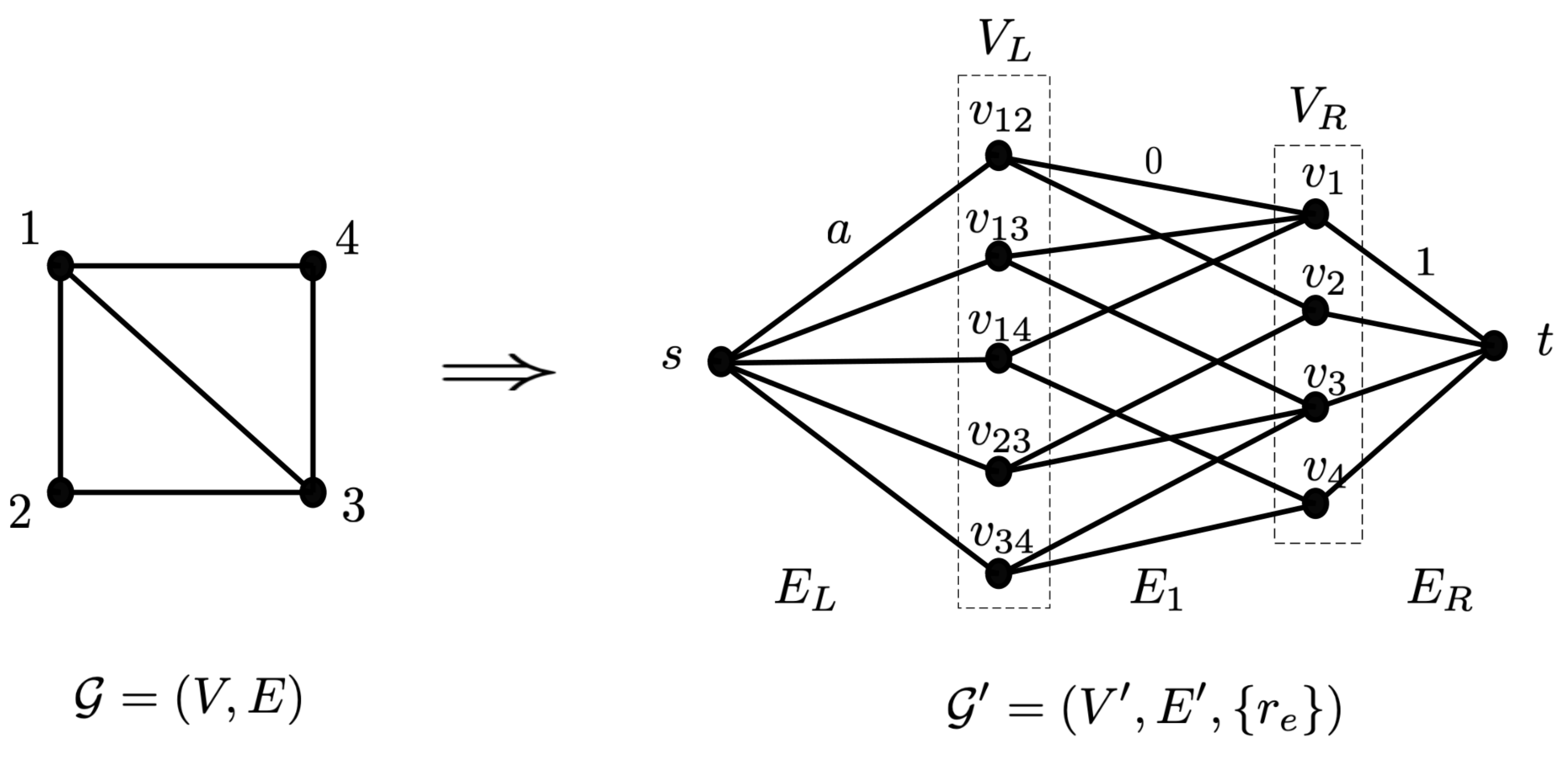} \hspace{0.4in}
\end{center}\vspace{-0.3cm}
\caption{An illustration of the reduced graph $\mathcal{G}'=(V', E', \{r_e\})$.}\label{fig:reduction}
\end{figure}



\begin{theorem}\label{thm:complexity}
The ERIP \eqref{eq:ERIP} is strongly NP-hard even if (i) the network is bipartite with diameter three, and (ii) the edges adjacent to the source node $s$ and the sink node $t$ are fixed and cannot be removed.  
\end{theorem}
\begin{proof}
We use reduction from the max-clique decision problem \cite{karp1972reducibility}, wherein an instance is given by an undirected graph $\mathcal{G}=(V, E)$ and an integer $r\in \mathbb{Z}_+$, and the goal is to determine whether $\mathcal{G}$ contains a clique of size $r$. We reduce it to an instance $\mathcal{G}'=(V', E', \{r_e\}_{e\in E'})$ of the ERIP, in which the network is constructed as above, and $\ell=2\big(m-{r \choose 2 }\big)$, $R_0=\frac{a}{{r \choose 2 }}+\frac{1}{r}$, where $a=n^4$ is a large constant. We also set 
\begin{align}\nonumber
r_e=\begin{cases}
0 \ &\mbox{if} \ e\in E_1,\\
a \ &\mbox{if} \ e\in E_L,\\
1 \ &\mbox{if} \ e\in E_R.\\
\end{cases}
\end{align}
Note that by condition (ii) only the edges in $E_1$ are allowed to be removed, and the edges in $E_L\cup E_R$ are fixed.

Let us first assume that $\mathcal{G}$ has a clique $C\subset V$ of $r$ vertices and let $\hat{E}=\{\{v_{ij},v_i\}: i,j\notin C \}$ be the set of all the $\ell=2(m-{r \choose 2 })$ edges corresponding to the nonclique edges. Then, we have $\mbox{R}_{{\rm eff}}(\mathcal{G}'\setminus \hat{E})=R_0$. The reason is that all the remaining edges $E'\setminus \hat{E}$ have zero resistance, in which case all the nodes $\{v_{ij}, v_i: i,j\in C\}$ can be shortcut and collapsed to a single vertex $u$. As a result, the $s-t$ effective resistance in $\mathcal{G}'\setminus \hat{E}$ is the same as that in a network of three nodes $s,u,t$, where there are ${r \choose 2}$ parallel edges of resistance $a$ between $s$ and $v$, and there are $r$ parallel edges of resistance $1$ between $v$ and $t$. Therefore, if $\mathcal{G}$ has a clique of size $r$, then there is a ``Yes" answer to the corresponding decision instance of the ERIP.

Conversely, suppose that $\mathcal{G}$ does not contain a clique of size $r$. We will show that for any edge cut $\hat{E}\subseteq E_1, |\hat{E}|=\ell$, we have $\mbox{R}_{{\rm eff}}(\mathcal{G}'\setminus \hat{E})<R_0$. To that end, we first note that if $E_1\setminus \hat{E}$ spans $q\ge {r \choose 2}+1$ vertices $\{v_{i_1j_1},\ldots,v_{i_qj_q}\}\subseteq V_L $, we can upper-bound $\mbox{R}_{{\rm eff}}(\mathcal{G}'\setminus \hat{E})$ as follows: let $f$ be a unit $s-t$ flow in $\mathcal{G}'\setminus \hat{E}$ that sends $\frac{1}{q}$-unit of flow over each of the edges $\{s,v_{i_1j_1}\},\ldots,\{s,v_{i_qj_q}\}$, and route it to $t$ arbitrarily by respecting the flow conservation constraints. By using Thomson's principle,  we can upper-bound $\mbox{R}_{{\rm eff}}(\mathcal{G}'\setminus \hat{E})$ using the cost of such a flow as
\begin{align}\label{eq:0-1-n-flow}
\mbox{R}_{{\rm eff}}(\mathcal{G}'\setminus \hat{E})&\leq \sum_{e\in E_L}a f^2_e +\sum_{e\in E_1\setminus \hat{E}}0\cdot f^2_e+\sum_{e\in E_R}1 \cdot f^2_e\cr 
&=q (\frac{a}{q^2})+\sum_{e\in E_R} f^2_e\leq \frac{a}{q}+1\cr 
&\leq \frac{a}{{r \choose 2}+1}+1<R_0-\frac{1}{r},
\end{align}
where the second inequality holds because $f$ is a unit flow such that $\sum_{e\in E_R} f^2_e\leq (\sum_{e\in E_R}f_e)^2=1$, and the last inequality holds because $a=n^4\ge {r \choose 2}({r \choose 2}+1)$. Therefore, the edges in $E_1\!\setminus\! \hat{E}$ span at most ${r \choose 2}$ vertices from $V_L$. On the other hand, $E_1\!\setminus\! \hat{E}$ contains exactly $2{r \choose 2}$ edges and each vertex in $V_L$ is incident to exactly two edges in $E_1$. As a result $E_1\!\setminus\! \hat{E}$ spans exactly ${r \choose 2}$ vertices from $V_L$, which also implies that $E_1\!\setminus\! \hat{E}$ spans at least $r+1$ nodes from $V_R$.

Let $C_1, C_2, \ldots, C_d$ for some $d\ge 1$ be the connected components of the induced graph $\mathcal{G'}[E_1\setminus \hat{E}]$, and let $n_k=|C_k\cap V_L|$ and $m_k=|C_k\cap V_R|$. From above, we have
\begin{align}\label{eq:m-n}
\sum_{k=1}^dn_k={r \choose 2},\ \ \ \  \ \ \ \sum_{k=1}^dm_k\ge r+1, \  \ \ \  n_k\leq {m_k \choose 2}\ \forall k,
\end{align}
where the last inequality holds because each connected component $C_k$ can be at most a clique with $m_k$ nodes. Therefore, $\mbox{R}_{{\rm eff}}(\mathcal{G}'\setminus \hat{E})$ can be computed by contracting each connected component $C_k$ to a single vertex $u_k$, and using series/parallel laws to get
\begin{align}\label{eq:final-effective}
\frac{1}{\mbox{R}_{{\rm eff}}(\mathcal{G}'\setminus \hat{E})}=\sum_{k=1}^{d}\frac{1}{\frac{a}{n_k}+\frac{1}{m_k}}.
\end{align}
If $d=1$, then, $n_1={r \choose 2}, m_1\ge r+1$, and we have
\begin{align}\label{eq:single-component}
\mbox{R}_{{\rm eff}}(\mathcal{G}'\setminus \hat{E})&=\frac{a}{n_1}+\frac{1}{m_1}\leq \frac{a}{{r \choose 2}}+\frac{1}{r+1}<R_0.
\end{align}
For $d\ge 2$ components, we can upper-bound $\mbox{R}_{{\rm eff}}(\mathcal{G}'\setminus \hat{E})$ using Thomson's principle for the equivalent contracted graph. In the contracted graph, consider a unit $s-t$ flow that sends $\frac{n_k}{{r \choose 2}}$ amount of flow over the branch $s\to u_k\to t$ with resistance $\frac{a}{n_k}+\frac{1}{m_k}$. Using Thomson's principle, we have
\begin{align}\label{eq:energy-flow-m-n}
\mbox{R}_{{\rm eff}}(\mathcal{G}'\setminus \hat{E})&\leq \sum_{k=1}^d(\frac{a}{n_k}+\frac{1}{m_k})\Big(\frac{n_k}{{r \choose 2}}\Big)^2=\frac{a}{{r \choose 2}}+\frac{1}{{r \choose 2}^2}\sum_{k=1}^d\frac{n^2_k}{m_k}\cr 
&\leq \frac{a}{{r \choose 2}}+\frac{1}{{r \choose 2}^2}\sum_{k=1}^d\frac{2n^2_k}{1+\sqrt{1+8n_k}},
\end{align}
where the equality holds because $\sum_{k=1}^dn_k={r \choose 2}$, and the last inequality uses $m_k\ge \frac{1+\sqrt{1+8n_k}}{2}$ due to \eqref{eq:m-n}.
In order to show $\mbox{R}_{{\rm eff}}(\mathcal{G}'\!\setminus\! \hat{E})<R_0$, it is enough to show that
\begin{align}\label{eq:max-value-flow}
\max_{\substack{n_k\ge1, \forall k\\ \sum n_k={r \choose 2}}}h(n_1,\ldots,n_d)=\sum_{k=1}^d\frac{2n^2_k}{1+\sqrt{1+8n_k}}< \frac{{r \choose 2}^2}{r},
\end{align} 
where the constraints $\{n_k\ge1, \forall k, \ \sum_{k=1}^dn_k={r \choose 2}\}$ define an integral polytope with exactly $d$ extreme points: each extreme point is obtain by setting $d-1$ of the variables to $1$, and the last variable to ${r \choose 2}-(d-1)$. Moreover, $h(\cdot)$ is a strictly convex function whose Hessian is a positive-definite diagonal matrix with the $k$th diagonal entry $\frac{\partial^2h}{\partial n^2_k}=\frac{2(1+6n_k)}{(1+8n_k)^{\frac{3}{2}}}>0$. Since a convex function achieves its maximum value at an extreme point of a polytope, using \eqref{eq:max-value-flow} and the above characterization of extreme points, we have
\begin{align}\label{eq:final-complexity-r}
\max_{\substack{n_k\ge1, \forall k\\ \sum_kn_k={r \choose 2}}}h(n_1,\ldots,n_d)&=\frac{d-1}{2}+\frac{2\big({r \choose 2}-(d-1)\big)^2}{1+\sqrt{1+8\big({r \choose 2}-(d-1)\big)}}\cr 
&\leq \frac{1}{2}+\frac{2\big({r \choose 2}-1\big)^2}{1+\sqrt{1+8\big({r \choose 2}-1\big)}}\cr 
&= \frac{1}{2}+\frac{\big({r \choose 2}-1\big)}{4}\big(\sqrt{1+8({r \choose 2}-1)}-1\big)\cr 
&<\frac{1}{2}+\frac{1}{2}\big({r \choose 2}-1\big)(r-1-\frac{1}{r})\cr 
&=\frac{{r \choose 2}^2}{r}-\frac{(3r+1)(r-2)}{4r}\leq \frac{{r \choose 2}^2}{r}, 
\end{align}
where the first inequity holds because the right side expression is a decreasing function of $d$,\footnote{The derivative of this function with respect to $d$ equals $\frac{3}{4}+\frac{1-3(4r(r-1)-8d+9)}{8\sqrt{4r(r-1)-8d+9}}<0, \forall d=2,\ldots,r$.} and achieves its maximum for $d=2$ components. The second inequality uses the relation $1+8\big({r \choose 2}-1\big)< (2r-1-\frac{2}{r})^2, \forall r\ge 2$, which establishes \eqref{eq:max-value-flow}.
\end{proof}

\subsection{Complexity of the CIP}\label{sec:scaling}
In this section, we extend the hardness result from the ERIP to the CIP using a series of reductions. We first replace the $0$-resistance edges in $E_1$ with edges of small resistance $n^{-2}$. We then show that scaling the conductances to form a symmetric stochastic matrix will preserve the NP-hardness. Finally, we use another polynomial scaling to show that the same hardness result holds for squared stochastic conductances.

\smallskip
\noindent
{\bf (I) Changing Resistance to Conductance:}  Any edge with polynomially bounded positive resistance $r_e$ can be replaced by an edge with polynomially bounded conductance $p_e=\frac{1}{r_e}$. Moreover, if an edge $e$ does not exist in $E'$, i.e., $r_e=\infty$, we have $p_e=0$. The only issue is with 0-resistance edges in $E_1$, which result in edges of unbounded conductance. However, this issue can be resolved by noting that the complexity result of Theorem \ref{thm:complexity} is robust with respect to small perturbations in the size of the resistances $\{r_e\}_{e\in E'}$. More precisely, if we replace each edge $e\in E_1$ by an edge of resistance $r_e=\frac{1}{n^2}$, then all the analysis in the proof of Theorem \ref{thm:complexity} carry over verbatim. The only difference is that $\mbox{R}_{\rm eff}(\mathcal{G}'\setminus \hat{E})$ in \eqref{eq:0-1-n-flow} can be upper-bounded by $\frac{a}{{r \choose 2}+1}+\frac{1}{n^2}+1$, which is again less than $R_0=\frac{a}{{r \choose 2}}+\frac{1}{r}$. Moreover, an extra term of size at most $\sum_{k=1}^d (2mn^2_k)/(n{r \choose 2})^2\leq \frac{2m}{n^2}$ will be added to the right side of \eqref{eq:energy-flow-m-n}. However, the effect of such a term to the final inequality in \eqref{eq:final-complexity-r} is at most $\frac{2m}{n^2}-\frac{(3r+1)(r-2)}{4r}$, which is strictly negative for any $r> 3$. As a result, the final inequality in \eqref{eq:final-complexity-r} still holds. Therefore, the perturbed instance with $\frac{1}{n^2}$-resistance edges (instead of $0$-resistance edges) is strongly NP-hard, where now the conductance of an edge in the perturbed instance belongs to the set $\{0, n^{-4}, 1, n^2\}$. Henceforth, we can only work with the perturbed instance.  

\smallskip
\noindent
{\bf (II) Symmetric Stochastic Conductance Matrices:} The restriction to symmetric stochastic conductance matrices does not make the problem any easier. The reason is that if we scale each conductance in the perturbed instance $\mathcal{G}'=(V', E')$ by the same factor of $n^{-3}$, the results of the previous section remain valid except that all the derivations are scaled by $n^{3}$. Using part (I), each edge in the scaled network has a conductance of at most $n^2\cdot n^{-3}=n^{-1}$. Therefore, for any node $i\in V'$, the sum of the conductances adjacent to node $i$ can be at most $\sum_{j: \{i,j\}\in E'}p_{ij}\leq 1$. Thus, by adding a self-loop to each node $i\in V'$ with conductance $p_{ii}=1-\sum_{j: \{i,j\}\in E'}p_{ij}$,\footnote{Note that adding self-loops has no effect on the effective resistance computations.} without loss of generality, one can assume that the input conductance matrix is a symmetric stochastic matrix. Based on these observations, we have the following theorem.

\begin{theorem}\label{thm:scaling-resistance}
Solving $\max_{|\hat{E}|\leq \ell}\mbox{{\rm R}}_{\rm eff}\big((Q \setminus\hat{E})^2\big)$ is strongly NP-hard, even for symmetric stochastic conductance inputs $Q$, defined over bipartite graphs of diameter three with fixed source/sink edges. In particular,  the CIP is strongly NP-hard.
\end{theorem}
\begin{proof}
From observations (I) and (II), we know that for symmetric stochastic conductance inputs $P$ with off-diagonal entries in $\{0,n^{-7}, n^{-3},n^{-1}\}$, it is strongly NP-complete to decide whether there exists an edge cut $|\hat{E}|\leq \ell$ such that $\mbox{R}_{\rm eff}(P\setminus\hat{E})\ge R_0n^3$, where the extra factor $n^3$ is because of step (II). We reduce this problem to the case of squared conductance matrices. To that end, let $\epsilon=n^{-21}$ be a small positive number, and define $Q=(1-\epsilon)I+\epsilon P$. Note that $Q$ is also a symmetric stochastic conductance matrix whose off-diagonal entries are polynomially bounded in terms of $n$. Moreover, both $P$ and $Q$ have the same set of edges so that any edge cut in $P$ is also a feasible edge cut in $Q$. We show that there is a ``Yes" answer to the ERIP with conductance matrix $P$, budget $\ell$, and resistance threshold $R_0n^3$, if and only if there is an edge cut $|\hat{E}|=\ell$ for the conductance matrix $Q$ such that $\mbox{R}_{\rm eff}\big((Q \setminus \hat{E})^2\big)\ge \frac{R_0n^3}{2\epsilon(1-\epsilon)}-n^{20}.$ 

First, we note that $Q^2=(1-\epsilon)^2I+2\epsilon(1-\epsilon)P+\epsilon^2P^2$. It is worth noting that $(1-\epsilon)^2I$ only changes the conductance of the self-loops and has no effect on the $s-t$ effective resistance. Since $\mbox{R}_{\rm eff}(2\epsilon(1-\epsilon)P)=\frac{\mbox{R}_{\rm eff}(P)}{2\epsilon(1-\epsilon)}$, we can write
\begin{align}\label{eq:lower-upper-n20}
\frac{\mbox{R}_{\rm eff}(P)}{2\epsilon(1-\epsilon)}- n^{20}\leq \mbox{R}_{\rm eff}(Q^2)\leq \frac{\mbox{R}_{\rm eff}(P)}{2\epsilon(1-\epsilon)},
\end{align}
where the upper bound is by Rayleigh's monotonicity law (Corollary \ref{cor:ray}) as $Q^2$ has more edges than $2\epsilon(1-\epsilon)P$ (i.e., the edges due to the term $\epsilon^2P^2$). The lower bound holds because all the entries of $\epsilon^2P^2$ are bounded above by $\epsilon^2$. Since $\frac{\mbox{R}_{\rm eff}(P)}{2\epsilon(1-\epsilon)}<\frac{n^{7}+n^3+n}{2\epsilon(1-\epsilon)}\leq n^7\epsilon^{-1}$, even if all the (at most) $n^6$ edges in $\epsilon^2P^2$ are added in parallel between $s$ and $t$ to the network $2\epsilon(1-\epsilon)P$, the  $s-t$ effective resistance cannot decrease by more than $n^{20}$. More precisely, $\mbox{R}_{\rm eff}(Q^2)$ is at least
\begin{align}\nonumber
\frac{\frac{\epsilon^{-2}}{n^6}\cdot \frac{\mbox{R}_{\rm eff}(P)}{2\epsilon(1-\epsilon)}}{\frac{\epsilon^{-2}}{n^6}+\frac{\mbox{R}_{\rm eff}(P)}{2\epsilon(1-\epsilon)}}&\ge \frac{\mbox{R}_{\rm eff}(P)}{2\epsilon(1-\epsilon)}-\Big(\frac{n^7\epsilon^{-1}}{\frac{\epsilon^{-2}}{n^6}+n^7\epsilon^{-1}}\Big)\frac{\mbox{R}_{\rm eff}(P)}{2\epsilon(1-\epsilon)}\cr 
&\ge \frac{\mbox{R}_{\rm eff}(P)}{2\epsilon(1-\epsilon)}-\Big(\frac{n^{13}\epsilon^{-1}}{\epsilon^{-2}+n^{13}\epsilon^{-1}}\Big)n^7\epsilon^{-1}\cr 
&\ge \frac{\mbox{R}_{\rm eff}(P)}{2\epsilon(1-\epsilon)}-n^{20}. 
\end{align}

Suppose that there is a ``Yes" answer to the ERIP with conductance matrix $P$. Then, there exists an edge cut $\hat{E}$ such that $\mbox{R}_{\rm eff}\big(P\setminus \hat{E}\big)\ge R_0n^3$. If we remove the same edge cut from $Q$, using the lower bound in \eqref{eq:lower-upper-n20} adapted for $Q\setminus \hat{E}$ and $P\setminus \hat{E}$, 
\begin{align}\nonumber
\mbox{R}_{\rm eff}\big((Q\setminus \hat{E})^2\big)&\ge \frac{\mbox{R}_{\rm eff}(P\setminus \hat{E})}{2\epsilon(1-\epsilon)}- n^{20}\ge \frac{R_0n^3}{2\epsilon(1-\epsilon)}-n^{20}.
\end{align}
Conversely, if the answer to the ERIP with input $P$ is ``No", using \eqref{eq:0-1-n-flow} and \eqref{eq:final-complexity-r} in the proof of Theorem \ref{thm:complexity}, for any edge cut $\hat{E}$, we have $\mbox{R}_{\rm eff}(P\setminus \hat{E})< (R_0-n^{-3})n^3$. Thus, using the upper bound in \eqref{eq:lower-upper-n20}, for any edge cut $\hat{E}$ in $Q$, we get
\begin{align}\nonumber
\mbox{R}_{\rm eff}\big((Q\setminus \hat{E})^2\big)\leq \frac{\mbox{R}_{\rm eff}(P\setminus \hat{E})}{2\epsilon(1-\epsilon)}< \frac{R_0n^3-1}{2\epsilon(1-\epsilon)}< \frac{R_0n^3}{2\epsilon(1-\epsilon)}-\frac{\epsilon^{-1}}{2}\leq \frac{R_0n^3}{2\epsilon(1-\epsilon)}-n^{20},
\end{align} 
which completes the reduction.
\end{proof}

\section{Hardness of Approximation for the ERIP and CIP}\label{sec:hardness-approx}
 
In this section, we provide strong inapproximability results for the ERIP and CIP. To that end, we borrow some ideas from \cite{chestnut2017hardness} to connect inapproximability of the ERIP to that of the network flow interdiction problem. The reduction is from the densest-$k$-subgraph (D$k$S) problem, wherein the goal is to find a $k$-vertex subgraph of $\mathcal{G}=(V,E)$, which has the maximum number of edges. The D$k$S is strongly NP-hard and does not admit any polynomial-time approximation scheme \cite{khot2006ruling}. Moreover, assuming the exponential time hypothesis \cite{impagliazzo2001complexity} (i.e., assuming nonexistence of a subexponential-time algorithm for solving 3SAT), there is no polynomial-time algorithm that approximates D$k$S to within $n^{\frac{1}{(\log \log n)^c}}$ factor of the optimum, where $c>0$ is a constant independent of $n$ \cite{manurangsi2017almost}. On the positive side, the best known polynomial-time approximation algorithm for D$k$S is due to \cite{bhaskara2010detecting} with an approximation ratio of $O(n^{\frac{1}{4}})$. Here, by an $\alpha$-approximation algorithm ($\alpha\ge 1$), we refer to a polynomial-time algorithm that satisfies $\mbox{OPT}(I)\leq \alpha \mbox{ALG}(I)$ for any input instance $I$, where OPT$(I)$ and ALG$(I)$ denote the objective values obtained by the maximum solution and the approximation algorithm, respectively.

Given an instance of the D$k$S, we reduce it to an instance $\mathcal{G}'=(V', E', \{r_e\}_{e\in E'})$ of the ERIP. The network $\mathcal{G}'$ is constructed as before (see Figure \ref{fig:reduction}) with edge resistances $r_e=1, \forall e\in E_R$ and $r_e=\delta, \forall e\in E_L\cup E_1$, where $\delta=n^{-4}$ is a small constant. As is shown in Lemma \ref{lemm:paralel}, by adding polynomially many parallel edges to $E_R$, any optimal edge cut must only remove edges from $E_L\cup E_1$. Henceforth, we assume that only the $\delta$-resistance edges in $E_L\cup E_1$ are removable, and the edges in $E_R$ are fixed. We note that a distinction between the reductions in Theorem \ref{thm:complexity} and the following theorem (Theorem \ref{thm:approximability}) is that the reduction in Theorem \ref{thm:approximability} holds under a less restrictive condition that allows removable source/sink edges. As a result, one can establish a stronger inapproximability result, as is shown below.

\begin{theorem}\label{thm:approximability}
Let $\epsilon\in(0, 1)$ be a constant. There is no $O(n^{\epsilon})$-approximation algorithm for the ERIP even for bipartite graphs of diameter three unless there is an $O(n^{4\epsilon})$-approximation algorithm for the D$k$S. In particular, assuming the exponential time hypothesis, the ERIP cannot be approximated within a factor better than $n^{\frac{1}{4(\log \log n)^{c}}}$ for some constant $c>0$.   
\end{theorem}
\begin{proof}
For any instance $\mathcal{G}=(V,E)$ of the D$k$S and $S\subseteq V$, we can associate an edge cut $\hat{E}_S\subseteq E'$ in the corresponding instance $\mathcal{G}'=(V', E')$ of the ERIP as: 
\begin{align}\label{eq:E-S}
\hat{E}_S=\big\{\{s,v_{ij}\}: \{i,j\}\in E[S] \big\}\cup\big\{\{v_i,v_{ij}\}: i\in S, \{i,j\}\in \delta(S)\big\}.
\end{align}
Note that removing the edges $\hat{E}_S$ from $\mathcal{G}'$ completely eliminates the vertices $\{v_i: i\in S\}$ from the $s-t$ effective resistance computation.\footnote{For the rest of the proof, we often abuse the notation $S$ to refer to a subset of nodes in $V$ or its corresponding vertices $\{v_i, i\in S\}$ in $V_{R}$.} As all the edges in $E_{L}\cup E_1$ have resistance $\delta<<1$, the $s-t$ effective resistance in $\mathcal{G}'\setminus \hat{E}_S$ can be well approximated by $|V_R|-|S|=|V\setminus S|$ parallel unit-resistance edges between $s$ and $t$. More precisely,
\begin{align}\label{eq:approx-Es}
\frac{1}{|V\setminus S|}\leq \boldsymbol{\mbox{R}}_{\rm eff}(\mathcal{G}'\setminus \hat{E}_S)\leq \frac{1}{|V\setminus S|}+2\delta,
\end{align}
Here, the lower bound is obtained by Rayleigh's monotonicty law when all the $\delta$-resistance edges in $\mathcal{G}'\setminus \hat{E}_S$ are replaced by $0$-resistance edges. The upper bound holds because in $\mathcal{G}'\setminus \hat{E}_S$, each vertex $v_j\in V\setminus S$ is connected to $s$ via a path of length two with resistance at most $2\delta$. Moreover, the number of edges in $\hat{E}_S$ equals to 
\begin{align}\nonumber
|\hat{E}_S|=|E[S]|+|\delta(S)|=m-|E[V\setminus S]|.
\end{align}
Thus, if $\hat{E}_{S}$ defines an optimal edge cut for ERIP on $\mathcal{G}'$, then $\mathcal{G}[V\setminus S]$ has the fewest vertices of any  $|E[V\setminus S]|$-edge subgraph. In other words, $\mathcal{G}[V\setminus S]$ must be the densest $|E[V\!\setminus\!S]|$-edge subgraph. On the other hand, as is shown in Lemma \ref{lemm-hardness}, for any edge cut $\hat{E}$ (not necessarily the optimal edge cut), one can find in polynomial time an edge cut of the form $\hat{E}_{S}$ given in \eqref{eq:E-S} such that $|\hat{E}_S|\leq |\hat{E}|$ and 
\begin{align}\label{eq:approx-effective-epsilon}
\boldsymbol{\mbox{R}}_{\rm eff}(\mathcal{G}'\setminus \hat{E}_S)\ge \boldsymbol{\mbox{R}}_{\rm eff}(\mathcal{G}'\setminus \hat{E})-4\delta.  
\end{align}


To establish the hardness result, let us assume that the ERIP on a graph of $n$ nodes admits an $\alpha(n)$-approximation algorithm. Consider an instance of the D$k$S on a graph $\mathcal{G}=(V,E)$ with $n$ nodes, and denote the number of edges in its densest $k$-vertex subgraph by $\ell^*$. Using \eqref{eq:approx-Es}, the maximum effective resistance for the corresponding instance $\mathcal{G}'=(V',E',\{r_e\})$ in the ERIP with budget $|\hat{E}^*|\leq m-\ell^*$ is at least $\boldsymbol{\mbox{R}}_{\rm eff}(\mathcal{G}'\setminus \hat{E}^*)\ge \frac{1}{k}$, where $\hat{E}^*$ denotes the optimal edge cut. Now, let us apply the approximation algorithm on $\mathcal{G}'$ and denote its output by $\hat{E}$. Using \eqref{eq:approx-effective-epsilon}, we can find in polynomial time an edge cut $\hat{E}_S$ that satisfies
\begin{align}\nonumber
\boldsymbol{\mbox{R}}_{\rm eff}(\mathcal{G}'\setminus \hat{E}_S)\ge \boldsymbol{\mbox{R}}_{\rm eff}(\mathcal{G}'\setminus \hat{E})-4\delta\ge \frac{\boldsymbol{\mbox{R}}_{\rm eff}(\mathcal{G}'\setminus \hat{E}^*)}{\alpha(|V'|)}-4\delta\ge \frac{1}{k\alpha(n^2)}-4\delta,
\end{align}
where the last inequality holds as $|V'|=n+m+2\leq n^2$. Thus, we can find in polynomial time an edge cut $\hat{E}_S$, such that
\begin{align}\nonumber
&|\hat{E}_S|=m-|E[V\!\setminus\!S]|\leq m-\ell^* \hspace{2.65cm}\Rightarrow\qquad |E[V\setminus S]|\ge \ell^*, \cr 
&\frac{1}{k\alpha(n^2)}-4\delta\leq \boldsymbol{\mbox{R}}_{\rm eff}(\mathcal{G}'\setminus \hat{E}_S)\leq \frac{1}{|V\setminus S|}+2\delta\qquad\Rightarrow\qquad \frac{1}{|V\setminus S|}\geq \frac{1}{k\alpha(n^2)}-6\delta.
\end{align}   
Therefore, we can find a subgragh $\mathcal{G}[V\setminus S]$ in polynomial time with at least $\ell^*$ edges and at most $(\frac{1}{k\alpha(n^2)}-6\delta)^{-1}$ vertices.  If we sample $k$ vertices uniformly from $\mathcal{G}[V\setminus S]$, the expected number of edges in that $k$-vertex random subgraph equals
\begin{align}\nonumber
\frac{{k \choose 2}}{{|V\setminus S| \choose 2}}|E[V\setminus S]|\ge \frac{k(k-1)}{|V\setminus S|(|V\setminus S|-1)}\ell^*\ge k(k-1)\ell^*(\frac{1}{k\alpha(n^2)}-6\delta)^2. 
\end{align}
This shows that there exists a $k$-vertex subgraph of $\mathcal{G}[V\setminus S]$ (and hence a $k$-vertex subgraph of $\mathcal{G}$) with at least $k(k-1)\ell^*(\frac{1}{k\alpha(n^2)}-6\delta)^2$ edges,  which can be found in polynomial time using a standard derandomization.  Therefore, if the ERIP admits an $\alpha(n)=O(n^{\epsilon})$ approximation algorithm for some $\epsilon<1$, then the D$k$S admits a polynomial-time approximation algorithm with approximation factor at most 
\begin{align}\nonumber
\frac{\ell^*}{k(k-1)\ell^*(\frac{1}{k\alpha(n^2)}-6\delta)^2}=\frac{k^2}{k(k-1)}\times\frac{\alpha^2(n^2)}{(1-6\delta k\alpha(n^2))^2}=O(\alpha^2(n^2))=O(n^{4\epsilon}),
\end{align}
where the equality holds because $6\delta k \alpha(n^{2})=O(\frac{n^{2\epsilon}}{n^3})=o(1)$.
\end{proof}

In the following theorem, we show the same inapproximability result for the CIP. The proof is similar to that in Section \ref{sec:scaling}, and we only sketch the main steps here.

\begin{theorem}\label{thm:CIP-approximation}
Let $\gamma\in(0, 1)$. There is no $n^{\gamma}$-approximation algorithm for the CIP unless there is an $O(n^{4\gamma})$-approximation algorithm for the D$k$S. In particular, assuming the exponential time hypothesis, the CIP cannot be approximated within a factor better than $n^{\frac{1}{4(\log \log n)^{c}}}$ for some constant $c>0$. 
\end{theorem} 
\begin{proof} 
First, we note that the inapproximability results in Theorem \ref{thm:approximability} hold even for bipartite graphs of diameter three with minimum edge resistance $\delta=n^{-4}$. Thus, without any issue, we can work with conductance (rather than resistance), where the conductance of each edge is bounded above by $n^{4}$. Also, using the same scaling argument as in Case II in Section \ref{sec:scaling}, we may assume that the input conductance matrix $P\in [0,1]^{n\times n}$ is stochastic. Now, as in Theorem \ref{thm:scaling-resistance} let $\epsilon=n^{-21}$ and $Q=(1-\epsilon) I+\epsilon P$. Moreover, let us assume that there exists an $\alpha(n)$-approximation algorithm for the CIP and denote its output for the conductance input $Q$ by $\hat{E}^a$. By the approximation guarantee, for any edge cut $\hat{E}, |\hat{E}|\leq \ell$ over $Q$ we have
\begin{align}\nonumber
\mbox{R}_{\rm eff}\big((Q\setminus \hat{E}^a)^2\big)\ge \frac{\mbox{R}_{\rm eff}\big((Q\setminus \hat{E})^2\big)}{\alpha(n)}.
\end{align}
Thus, using \eqref{eq:lower-upper-n20} adapted for $P\setminus \hat{E}^a$ and $Q\setminus \hat{E}^a$, we get
\begin{align}\nonumber
\frac{\mbox{R}_{\rm eff}(P\setminus \hat{E}^a)}{2\epsilon(1-\epsilon)}\ge \mbox{R}_{\rm eff}\big((Q\setminus \hat{E}^a)^2\big)\ge \frac{\mbox{R}_{\rm eff}\big((Q\setminus \hat{E})^2\big)}{\alpha(n)}\ge \frac{1}{\alpha(n)}\big(\frac{\mbox{R}_{\rm eff}(P\setminus \hat{E})}{2\epsilon(1-\epsilon)}- n^{20}\big).
\end{align}   
Since $\epsilon=n^{-21}$, for any edge cut $\hat{E}$, we obtain
\begin{align}\nonumber
\mbox{R}_{\rm eff}(P\setminus \hat{E}^a)\ge \frac{\mbox{R}_{\rm eff}(P\setminus \hat{E})-2\epsilon(1-\epsilon) n^{20}}{\alpha(n)}\ge \frac{\mbox{R}_{\rm eff}(P\setminus \hat{E})-2n^{-1}}{\alpha(n)}\ge \frac{\mbox{R}_{\rm eff}(P\setminus \hat{E})}{2\alpha(n)},
\end{align}  
where the last inequality holds  because $\mbox{R}_{\rm eff}(P\setminus \hat{E})>>\frac{2}{n}$. (Recall that $P$ is a stochastic conductance matrix associated with a bipartite graph of diameter three with very small off-diagonal entries.) This shows that if the CIP admits an $\alpha(n)$-approximation, then one can approximate the ERIP over bipartite graphs of diameter three with conductance matrix $P$ within a factor of at most $2\alpha(n)$. 
\end{proof}

\section{Approximation Algorithms for the ERIP and CIP}\label{sec:quadratic}

As we showed in the previous section, the ERIP and CIP are NP-hard problems and are unlikely to admit approximation algorithms to a nearly polynomial factor.  In this section, we consider these problems and develop algorithms to approximate their optimal solutions.  

\subsection{An Approximation Algorithm for the ERIP}

In this section, we develop a polynomial-factor approximation algorithm (Algorithm \ref{alg:ERIP}) for the ERIP with arbitrary edge resistances. The algorithm removes high resistance edges using repeated application of the min $s-t$ cut problem to ensure that each path in the interdicted network contains at least one edge of high resistance.

\begin{algorithm}[H]
\caption{An Approximation Algorithm for the ERIP}\label{alg:ERIP}
\noindent 
{\bf Input:} A resistance network $\mathcal{G}=(V, E,\{r_e\}_{e\in E})$ with $|V|=n$ nodes, $|E|=m$ edges, two terminals $s,t\in V$, and interdiction budget $\ell$. 

-- Sort the edges based on their resistances such that $r_{e_1}\leq r_{e_2}\leq \ldots\leq r_{e_m}$. 

-- For $i=1,2,\ldots,m$, 
\begin{itemize} 
\item Let $\mathcal{G}_i=(V, \{e_1,\ldots,e_i\})$ be the unweighted graph obtained from the first $i$ edges with smallest resistance.  
\item Let $E_i$ be the set of edges in the unweighted min $s-t$ cut in $\mathcal{G}_i$.  
\item Let $k$ be the first time such that $|E_{k+1}|=\ell+1$.  Output $E^a:=E_k$ and stop.
\end{itemize}
\end{algorithm}

\begin{definition}
We say $E^a$ is an $\alpha$-approximate solution for the ERIP if 
\begin{align}\nonumber
{\rm R}_{\rm eff}(E\setminus E^a)\ge \frac{1}{\alpha}\max_{|\hat{E}|\leq \ell}{\rm R}_{\rm eff}(E\setminus \hat{E}).
\end{align}
\end{definition}

\begin{theorem}
For a network of $n$ nodes, $m$ edges, and arbitrary edge resistances, Algorithm \ref{alg:ERIP} returns an $nm$-approximate solution to the ERIP in time $O(n^2m^{\frac{3}{2}})$.
\end{theorem}
\begin{proof}
Given a resistance network $\mathcal{G}=(V,E,\{r_e\}_{e\in E})$,  let $\Phi(E)$ be the minimum over all $s-t$ paths of the maximum resistance of an edge on the path, i.e., 
\begin{align}\nonumber
\Phi(E)=\min_{P\in \mathcal{P}_{st}(E)}\max_{e\in P} \ r_e,
\end{align}
where $\mathcal{P}_{st}(E)$ denotes the set of all $s-t$ paths that are supported over the edge set $E$.  We first argue that Algorithm \ref{alg:ERIP} returns an optimal solution to the $\Phi$-value interdiction problem, i.e., $E^a=\argmax_{|\hat{E}|\leq \ell}\Phi(E\setminus \hat{E})$, where $E^a$ denotes the solution returned by Algorithm \ref{alg:ERIP}. To show that, let $e_{k+1}$ be the last edge that is processed by the algorithm before its termination.  It means that the unweighted $s-t$ min cuts in networks $\mathcal{G}_k=(V,\{e_i\}_{i=1}^{k})$ and $\mathcal{G}_{k+1}=(V,\{e_i\}_{i=1}^{k+1})$ contain $\ell$ and $\ell+1$ edges, respectively.  In particular, the algorithm's output $E^a$ is the unweighted $s-t$ min cut for $\mathcal{G}_k=(V,\{e_i\}_{i=1}^{k})$.  Therefore, every $s-t$ path in $(V, E\setminus E^a)$ must contain at least one edge from $\{e_{k+1},\ldots,e_m\}$, and hence $\Phi(E\setminus E^a)\ge r_{e_{k+1}}$.  On the other hand, since the unweighted $s-t$ min cut in $\mathcal{G}_{k+1}=(V,\{e_i\}_{i=1}^{k+1})$ has $\ell+1$ edges, any interdiction set $\hat{E}$ that removes at most $\ell$ edges will leave at least one $s-t$ path whose edges all belong to $\{e_i\}_{i=1}^{k+1}$, and thus $\Phi(E\setminus \hat{E})\leq r_{e_{k+1}}$.  Therefore, we have 
\begin{align}\label{eq:phi-Ea}
\Phi(E\setminus E^a)=\max_{|\hat{E}|\leq \ell }\Phi(E\setminus \hat{E}).
\end{align}

Let us now denote the optimal interdiction set to the ERIP by $E^*$.  Using the definition of $\Phi(E\setminus E^*)$, there exists at least one $s-t$ path in the network $(V, E\setminus E^*)$, such that every edge on that path has resistance at most $\Phi(E\setminus E^*)$.  Since a path can have at most $n$ edges, we get $\mbox{R}_{\rm eff}(E\setminus E^*) \leq  n\Phi(E\setminus E^*)$. On the other hand,  using the definition of $\Phi(E\setminus E^a)$, every $s-t$ path in $(V, E\setminus E^a)$ has at least one edge $e$ with resistance $r_e\ge \Phi(E\setminus E^a)$.  Since $(V, E\setminus E^a)$ has at most $m$ edges,  $\mbox{R}_{\rm eff}(E\setminus E^a)$ can be lower-bounded by the effective resistance of at most $m$ parallel $s-t$ paths, where we just showed that each such path has a resistance of at least $\Phi(E\setminus E^a)$. Therefore, $\mbox{R}_{\rm eff}(E\setminus E^a)\ge \frac{\Phi(E\setminus E^a)}{m}$. Now, using \eqref{eq:phi-Ea} we can write
\begin{align}\nonumber
\mbox{R}_{\rm eff}(E\setminus E^a)\ge \frac{\Phi(E\setminus E^a)}{m}\ge \frac{\Phi(E\setminus E^*)}{m}\ge \frac{\mbox{R}_{\rm eff}(E\setminus E^*)}{nm}.
\end{align}  
Finally, we note that Algorithm \ref{alg:ERIP} terminates when at most all the $m$ edges are added one by one (in which case $E_m\ge \ell+1$ due to Assumption \ref{ass:l-connectivity}). Moreover, for each edge $e_i$, Algorithm \ref{alg:ERIP} needs to solve one max $s-t$ flow problem in order to find the min $s-t$ cut $E_i$. As each max $s-t$ flow problem can be solved in $O(n^2m^{\frac{1}{2}})$, the running time of Algorithm \ref{alg:ERIP} is at most $O(n^2m^{\frac{3}{2}})$.          
\end{proof}

\subsection{An Approximation Algorithm for the CIP}

Extending the approximation algorithm given in the previous section to the CIP is more complicated and faces additional challenges.  For instance, the initial vector $x_0$ that is part of the input to the CIP affects all the nodes in the network (rather than only nodes $s$ and $t$). Therefore,  one must deal with a generalized network flow problem with multiple sources and sinks.  In this section, we instead take a different approach by providing a quadratic program formulation for the CIP. This alternative formulation provides new insights on how to obtain good approximate solutions for the CIP and makes interesting connections between the CIP, spectral connectivity, and power dissipation in electric networks. 

Let us consider an arbitrary symmetric stochastic conductance matrix $P$. We represent a feasible interdiction set $\hat{E}$ using its (complement) characteristic vector $y\in\{0,1\}^m$, i.e., $y_{ij}=y_{ji}=0$ if $\{i,j\}\in \hat{E}$ and $y_{ij}=y_{ji}=1$, otherwise. As a feasible interdiction set $\hat{E}$ can break at most $\ell$ edges, thus $y'\boldsymbol{1}\ge m-\ell$.  Moreover, we can rewrite the interdicted matrix $P\setminus \hat{E}$ in terms of the decision variable $y$ as 
\begin{align}\label{eq:p-y}
[P(y)]_{ij}=\begin{cases}p_{ij}y_{ij}  &\mbox{if } i\neq j, p_{ij}>0, \\
1-\sum_{k\neq i}p_{ik}y_{ik}  &\mbox{if } i=j.
 \end{cases}
\end{align}
We can now formulate the CIP using a quadratic program with linear constraints. 

\begin{lemma}\label{thm:quadratic-CIP}
Let $\mathcal{L}(y)=I-[P(y)]^2$, and consider any optimal solution $(y^*, u^*)$ to the following quadratic program 
\begin{align}\label{eq:u-perpendic-CIP}
&\min u'\big(\mathcal{L}(y)+\frac{J}{n}\big)u\cr 
&\qquad u'x_0=1,\cr 
\mbox{s.t.}&\qquad y'\boldsymbol{1}\ge m-\ell, \cr 
&\qquad y\in [0, 1]^m, u\in \mathbb{R}^n.
\end{align}
Then, $y^*$ is complement of the incidence vector of the optimal edge cut $\hat{E}^*$ in CIP. 
\end{lemma} 
\begin{proof}
Using Lemma \ref{lemm:eff-cast}, solving the CIP is equivalent to solving
\begin{align}\label{eq:more-general-ERIP}
\max_{|\hat{E}|\leq \ell} \ x_0'\Big(I-(P\!\setminus\!\hat{E})^2+\frac{J}{n}\Big)^{-1}x_0=\max_{\substack{y\in \{0, 1\}^m\\ y'\boldsymbol{1}\ge m-\ell}} \ x_0'\Big(\mathcal{L}(y)+\frac{J}{n}\Big)^{-1}x_0.
\end{align}
Thus, we only need to show that the optimal solution $y^*$ to the quadratic program \eqref{eq:u-perpendic-CIP} can be obtained by solving \eqref{eq:more-general-ERIP}.  As is shown in Lemma \ref{lemm:concave}, for any vector $u$, the objective function $u'\big(\mathcal{L}(y)+\frac{J}{n}\big)u$ is concave with respect to $y$. Therefore, minimizing $u'\big(\mathcal{L}(y)+\frac{J}{n}\big)u$ over the integral polytope $Y=\{y\in [0,1]^{m}: y'\boldsymbol{1}\ge m-\ell\}$ would deliver a binary vector $y^*$. (Note that the constraint set $Y$ is independent of the $u$ variable.) Thus, without loss of generality, we can drop the binary constraints on $y$ to obtain the following equivalent program:
\begin{align}\label{eq:u-into-objective}
\min \{u'\big(\mathcal{L}(y)+\frac{J}{n}\big)u: \ u'x_0=1,\ u\in \mathbb{R}^n, y\in Y\}.
\end{align} 

As the matrix $\mathcal{L}(y)+\frac{J}{n}$ is positive-definite (and hence invertible), if we define $U=\{u\in \mathbb{R}^n: u'x_0=1\}$, the optimization problem \eqref{eq:u-into-objective} can be written as
\begin{align}\label{eq:u-y-separate}
\!\!\!\min_{y\in Y, u\in U} \!u'\big(\mathcal{L}(y)+\frac{J}{n}\big)u=\min_{y\in Y}\big\{\!\min_{u\in U} u'\big(\mathcal{L}(y)+\frac{J}{n}\big)u\big\},
\end{align}
where the equality holds because the constraint sets $U$ and $Y$ are uncoupled. Now, for any fixed $y\in Y$, the inner minimization $\min_{u\in U}u'\big(\mathcal{L}(y)+\frac{J}{n}\big)u$ can be solved in a closed-form using Lagrangian duality. We note that this inner minimization is a positive-definite quadratic program with linear constraint $u'x_0=1$, and hence has zero-duality gap. Therefore, for any fixed $y$, if we define the Lagrangian function $L(u,\lambda)=u'\big(\mathcal{L}(y)+\frac{J}{n}\big)u-\lambda(u'x_0-1)$, the optimal primal-dual solutions to the inner minimization are given by 
\begin{align}\nonumber 
&u^*=\argmin_{u\in \mathbb{R}^n}L(u,\lambda^*)=\frac{\lambda^*}{2} \big(\mathcal{L}(y)+\frac{J}{n}\big)^{-1}x_0,\cr 
&\lambda^*=\argmax_{\lambda\in \mathbb{R}}L(u^*,\lambda)=\frac{2}{x'_0(\mathcal{L}(y)+\frac{J}{n})^{-1}x_0}.
\end{align}
By combining the above relations, we get 
\begin{align}\label{eq:optimal-u}
u^*=\frac{(\mathcal{L}(y)+\frac{J}{n})^{-1}x_0}{x'_0(\mathcal{L}(y)+\frac{J}{n})^{-1}x_0}.
\end{align}
Thus, we can write
\begin{align}\nonumber
\min_{u\in U}u'\big(\mathcal{L}(y)+\frac{J}{n}\big)u&=u^{*'}\big(\mathcal{L}(y)+\frac{J}{n}\big)u^*=\frac{1}{x'_0(\mathcal{L}(y)+\frac{J}{n})^{-1}x_0}.
\end{align}
Finally, using the above relation together with \eqref{eq:u-y-separate}, we obtain
\begin{align}\nonumber
y^*&=\argmin_{y\in Y}\big\{\min_{u\in U} u'\big(\mathcal{L}(y)+\frac{J}{n}\big)u\big\}\cr 
&=\argmin_{y\in Y}\frac{1}{x'_0(\mathcal{L}(y)+\frac{J}{n})^{-1}x_0}\cr 
&=\argmax_{y\in Y}x'_0(\mathcal{L}(y)+\frac{J}{n})^{-1}x_0,
\end{align}    
which is precisely the solution to the CIP \eqref{eq:more-general-ERIP}.
\end{proof}

\begin{remark}
In fact, the program \eqref{eq:u-perpendic-CIP} can be written in an equivalent form of 
\begin{align}\label{eq:second-eigenvalue}
\min \big\{u'\mathcal{L}(y)u: u'x_0=1, u'\boldsymbol{1}=0, y'\boldsymbol{1}\ge m-\ell, y\in [0, 1]^m, u\in \mathbb{R}^n\big\}.
\end{align}
Interestingly, if the initial vector $x_0$ could be set to $u$, then by Courant-Fischer Theorem, the optimization \eqref{eq:second-eigenvalue} would become $\min \{\lambda_2\big(\mathcal{L}(y)\big): y\in Y\}$, where $\lambda_2\big(\mathcal{L}(y)\big)$ denotes the second smallest eigenvalue of the Laplacian matrix $\mathcal{L}(y)$. In that case,  the CIP would reduce to minimizing the second smallest eigenvalue value of the Laplacian matrix $I-P^2$ by removing at most $\ell$ edges.
\end{remark}

Using Lemma \ref{thm:quadratic-CIP}, a natural approach to solve the CIP is to use an iterative algorithm based on the block-coordinate descent (BCD) \cite{razaviyayn2013unified}. At each iteration $\tau=1,2,\ldots$, the algorithm fixes one variable and optimizes the objective function with respect to the second variable. For a fixed \emph{network} variable $y^{\tau}$, the optimal voltage variable $u$ is obtained from expression \eqref{eq:optimal-u} given in Lemma \ref{thm:quadratic-CIP}. Unfortunately, for a fixed \emph{voltage} variable $u^{\tau}$, the objective function $u'^{\tau}\big(\mathcal{L}(y)+\frac{J}{n}\big)u^{\tau}$ is concave with respect to $y$,  which, in general, it could be hard to minimize over the polytope $Y$. Instead, we update the network variable $y$ using an inexact BCD method by minimizing an upper approximation of the objective function $u'^{\tau}\big(\mathcal{L}(y)+\frac{J}{n}\big)u^{\tau}$ at the current network variable $y=y^\tau$. The overall procedure is summarized in Algorithm \ref{alg:iterative}. 

\begin{algorithm}[t]
\caption{An Adaptive Iterative Approximation Algorithm for the CIP}\label{alg:iterative}

\noindent
{\bf Input:} Initial vector $x_0$,  budget $\ell$, and a symmetric stochastic conductance matrix $P=(p_{ij})$ with $m$ edges $E$ (excluding the self-loops). 

\smallskip
For $\tau=0,1,2,\ldots$,  and an arbitrary initial vector $y^0\in \{0,1\}^m$,
\begin{itemize}
\item Let $u^{\tau}=\frac{(\mathcal{L}(y^{\tau})+\frac{J}{n})^{-1}x_0}{x'_0(\mathcal{L}(y^{\tau})+\frac{J}{n})^{-1}x_0}$, where $\mathcal{L}(y^{\tau})=I-[P(y^{\tau})]^2$.
\item Let $f(u,y)=u'\big(\mathcal{L}(y)+\frac{J}{n}\big)u$. For all the edges $\{i,j\}\in E, i\neq j$, sort gradients 
\begin{align}\label{eq:network-update}
\nabla_{y_{ij}} f(u^{\tau},y^{\tau})&=2p_{ij}(u_i^{\tau}-u_j^{\tau})^2\big([P(y^{\tau})]_{ii}+[P(y^{\tau})]_{jj}-2[P(y^{\tau})]_{ij}\big)\cr 
&\qquad+2p_{ij}\sum_{k\neq i,j}\big((u^{\tau}_k-u^{\tau}_j)^2-(u^{\tau}_k-u^{\tau}_i)^2\big)\big([P(y^{\tau})]_{ik}-[P(y^{\tau})]_{jk}\big),
\end{align}and set $y_{ij}^{\tau+1}=0$ for the $\ell$ links of the highest gradients, and $y_{ij}^{\tau+1}=1$ for the remaining $m-\ell$ edges. 
\item Output $y^{\tau}$ as an approximate solution if $y^{\tau}=y^{\tau+1}$.     
\end{itemize}
\end{algorithm}

\begin{definition}
Let $f(u,y)=u'\big(\mathcal{L}(y)+\frac{J}{n}\big)u$, $Y=\{y\in [0,1]^{m}: y'\boldsymbol{1}\ge m-\ell\}$, and $U=\{u\in \mathbb{R}^n: u'x_0=1\}$. We say $(u^o,y^o)\in U\times Y$ is a first-order stationary point for the CIP if for any $(d_1,d_2)$ such that $(u^o+d_1,y^o+d_2)\in U\times Y$, we have 
\begin{align}\nonumber
\lim_{\lambda\downarrow 0}\frac{f(u^o+\lambda d_1,y^o+\lambda d_2)-f(u^o,y^o)}{\lambda}\ge 0.
\end{align}  
\end{definition} 

In other words, $(u^o,y^o)$ is a first-order stationary point if the directional derivative of the CIP objective function along any feasible direction is nonnegative.

\begin{theorem}
Algorithm \ref{alg:iterative} converges to a first-order stationary point for the CIP after finitely many iterations.
\end{theorem}
\begin{proof}
Let $f(u,y)=u'\big(\mathcal{L}(y)+\frac{J}{n}\big)u$. Given $(u^{\tau},y^{\tau})$ at iteration $\tau$, let us define
\begin{align}\label{eq:linearized-obj}
L_{\tau}(y):=f(u^{\tau},y^{\tau})+(y-y^{\tau})' \nabla_y f(u^{\tau},y^{\tau}),\ \forall y\in Y,
\end{align}
which can be viewed as a linearization of the objective function $f(u^{\tau},y)$ at the point $y=y^{\tau}$. Thanks to concavity of $f(u^{\tau},y)$ as a function of $y$ (Lemma \ref{lemm:concave}), we have
\begin{align}\label{eq:BCD-upper-approx}
f(u^{\tau},y)\leq L_{\tau}(y),\ \forall y\in Y.
\end{align}  
Minimizing the linearized function \eqref{eq:linearized-obj} over $Y=\{y\in [0,1]^{m}: y'\boldsymbol{1}\ge m-\ell\}$ admits a simple closed-form solution: sort the gradient components $\nabla_{y_{ij}} f(u^{\tau},y^{\tau})$ for all the links $\{i,j\}\in E$, and set $y_{ij}=0$ for the $\ell$ links of the highest gradient components. As is shown in Lemma \ref{lemm:gradient-expansion}, the gradient $\nabla_{y_{ij}} f(u^{\tau},y^{\tau})$ equals the right-hand side expression in \eqref{eq:network-update}, which gives us exactly the network update rule in Algorithm \ref{alg:iterative}. Therefore, $y^{\tau+1}=\argmin_{y\in Y} L_{\tau}(y)$ and we can write
\begin{align}\nonumber
f(u^{\tau+1},y^{\tau+1})&=\min_{u\in U}f(u,y^{\tau+1}) \leq f(u^{\tau},y^{\tau+1})\cr 
&\leq L_{\tau}(y^{\tau+1})=\min_{y\in Y} L_{\tau}(y)\leq L_{\tau}(y^{\tau})= f(u^{\tau},y^{\tau}),
\end{align}
where the second inequality is due to \eqref{eq:BCD-upper-approx}. In particular, due to strict convexity and concavity of $f(u,y)$ with respect to $u$ and $y$, respectively, at least one of the inequalities in the above expression is strict unless $y^{\tau+1}=y^{\tau}$. This shows that after each major iteration, the objective function $f(u,y)$ strictly decreases. As a result, no pair of points $(u^{\tau},y^{\tau})$ will be repeated twice during the execution of Algorithm \ref{alg:iterative}. Since $y^{\tau}$ is a binary vector that belongs to ${m \choose \ell}$ many extreme points of $Y$, the algorithm will terminate after finitely many iterations $\tau^*$ to some point $(u^{\tau^*},y^{\tau^*})$.  

Finally, using Taylor expansion and differentiability of $f(u,y)$ over the convex and compact set $U\times Y$, for any $(d_1,d_2)$ such that $(u^{\tau^*}\!\!\!+d_1,y^{\tau^*}\!\!\!+d_2)\in U\times Y$, we have
\begin{align}\label{eq:directional-derivative}
\lim_{\lambda\downarrow 0}\frac{f(u^{\tau^*}\!\!+\!\lambda d_1,y^{\tau^*}\!\!+\!\lambda d_2)-f(u^{\tau^*},y^{\tau^*})}{\lambda}&=d'_1\nabla_u f(u^{\tau^*},y^{\tau^*})+d'_2\nabla_y f(u^{\tau^*},y^{\tau^*})\cr 
&=d'_1\nabla_u f(u^{\tau^*},y^{\tau^*})+d'_2\nabla_y L_{\tau^*}(y^{\tau^*}).
\end{align}
Since $y^{\tau^*}=\argmin_{y\in Y}L_{\tau^*}(y)$ and $y^{\tau^*}+d_2\in Y$, we have
\begin{align}\label{eq:linear-y-final}
d'_2\nabla_y L_{\tau^*}(y^{\tau^*})=L_{\tau^*}(y^{\tau^*}\!+\!d_2)-L_{\tau^*}(y^{\tau^*})\ge 0.
\end{align} 
Moreover, $f(u,y^{\tau^*})$ is a convex function of $u$ such that $u^{\tau^*}=\argmin_{u\in U}f(u,y^{\tau^*})$ and $u^{\tau^*}+d_1\in U$. Thus, using the optimality condition for constrained convex optimization we must have $d'_1\nabla_u f(u^{\tau^*},y^{\tau^*})\ge 0.$ Substituting this relation and \eqref{eq:linear-y-final} into \eqref{eq:directional-derivative} shows that $(u^{\tau^*},y^{\tau^*})$ must be a first-order stationary point for the CIP.         
\end{proof} 

\begin{corollary}\label{rem:potential-theoretic}[Potential-Theoretic Algorithm] If at iteration $\tau$ we use a non-adaptive upper approximation given by the linearization of $f(u^{\tau},y)$ at the origin $y=0$, i.e., $L_{\tau}(y)=f(u^{\tau},0)+y' \nabla_y f(u^{\tau},0)$, then the update rule in \eqref{eq:network-update} degenerates to $\nabla_{y_{ij}} f(u^{\tau},0)=4p_{ij}(u^{\tau}_i-u^{\tau}_j)^2$. In that case, the network update at time $\tau+1$ simplifies as follows: sort power dissipations $p_{ij}(u^{\tau}_i-u^{\tau}_j)^2, \forall \{i,j\}\in E$, and break $\ell$ edges of the highest power dissipation. This is exactly the \emph{potential-theoretic} algorithm that was developed in \cite{khanafer2015robust}, which can be obtained as a special case of Algorithm \ref{alg:iterative}.
\end{corollary}

\section{Numerical Experiments}\label{sec:numerical}
In this section, we evaluate the performance of Algorithm \ref{alg:iterative} on various networks of different sizes. We show that the stationary solution obtained at the end of Algorithm \ref{alg:iterative} provides a high-quality approximate solution for the CIP both in terms of the objective value as well as the running time. We also compare the performance of Algorithm \ref{alg:iterative} with the non-adaptive potential-theoretic algorithm (Corollary \ref{rem:potential-theoretic}) and show its outperformance on most of the instances. In our experiments, we consider the following set of networks:  
\begin{itemize}
\item $K_n$: Complete graphs of $n$ nodes.
\item $K_{\frac{n}{2},\frac{n}{2}}$: Complete bipartite graphs with two part sizes $\lfloor \frac{n}{2} \rfloor$ and $\lceil\frac{n}{2}\rceil$ nodes.
\item $D_4$: $4$-Regular graphs in which $n$ nodes are arranged around a cycle and each node is connected to two of its immediate nodes on the left and on the right. 
\item $E_{\frac{1}{2}}$: Erdos-Renyi graphs on $n$ nodes with edge emergence probability $p=\frac{1}{2}$.\footnote{If necessary, we regenerate such graphs to assure that the output graph is $\ell$-edge connected.} 
\end{itemize} 

Next, we construct a symmetric stochastic conductance matrix $P$ associated with each network. To that end, we randomly assign an integer weight $w_{ij}=w_{ji}\sim \mbox{Unif}\{1,2,\ldots,10\}$ to each edge $\{i,j\}$ in the network, and set $w_{ij}=0$ if there is no edge between $i$ and $j$. We then normalize the edge wights using the well-known Metropolis matrices by setting
\begin{align}\nonumber
P_{ij}=\begin{cases}
\frac{w_{ij}}{\max\{\sum_{k=1}^nw_{ik},\sum_{k=1}^nw_{jk}\}}, & \mbox{if $i\neq j$}\\
1-\sum_{k\neq i}P_{ik}, & \mbox{if $i=j$}.
\end{cases}
\end{align}

Finally, we choose the initial vector $x_0\in \{0,1\}^n$ such that $[x_0]_i=0$ if $i$ is odd and $[x_0]_i=1$ if $i$ is even. In our first experiment, we set the edge budget to $\ell=3$ and increase the number of nodes for each of the above networks up to $n=25$. The results are reported in Figure \ref{fig:complete}. The optimal CIP objective value for each of these networks is shown by the greed dashed curve, and the objective value of the solution returned by Algorithm \ref{alg:iterative} and the potential-theoretic algorithm are depicted by the red and blue curves, respectively. As it can be seen from Figure \ref{fig:complete}, in most of the instances, Algorithm \ref{alg:iterative} outperforms the potential-theoretic algorithm, and its overall performance is very close to that of the optimal algorithm. In particular, the gap between the optimal curve and the potential-theoretic algorithm shows the suboptimality of this algorithm that was erroneously argued in \cite{khanafer2015robust} to be the globally optimal algorithm.

\begin{figure}[t]
\vspace{3.7cm}
\begin{center}
\includegraphics[totalheight=.3\textheight,
width=.45\textwidth,viewport=270 0 820 550]{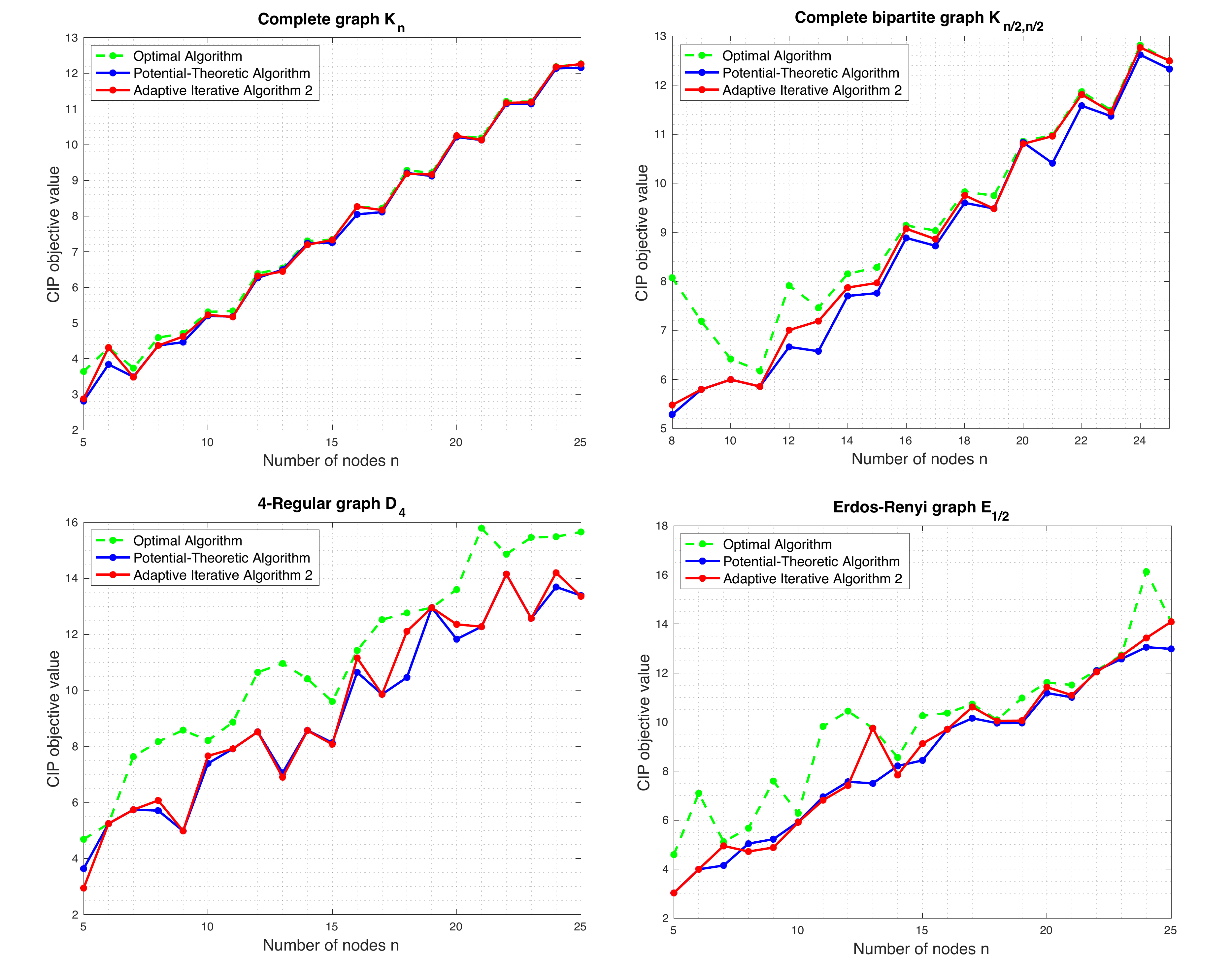} \hspace{0.4in}
\end{center}\vspace{-0.3cm}
\caption{Comparison between the performance of the optimal algorithm, Algorithm \ref{alg:iterative}, and the potential-theoretic algorithm for various network structures.}\label{fig:complete}
\end{figure}

In Figure \ref{fig:iterate} (left-side), we have compared the performance of Algorithm \ref{alg:iterative} and the potential-theoretic algorithm for a fixed number of nodes and different edge budgets. We have simulated the performance of those algorithms for the networks $E_{\frac{1}{2}}$ and $K_n$ with $n=70$ nodes while changing the edge budget from $\ell=1$ to $\ell=25$. As it can be seen, Algorithm \ref{alg:iterative} achieves a substantially better objective value than the potential-theoretic algorithm. Finally, in Figure \ref{fig:iterate} (right-side), we have evaluated the maximum number of iterations before Algorithm \ref{alg:iterative} terminates. To that end, we have fixed the edge budget to $\ell=20$ and increased the number of nodes from $n=50$ to $n=150$ for three different networks $K_n$, $K_{\frac{n}{2},\frac{n}{2}}$, and $E_{\frac{1}{2}}$. Moreover, for each fixed value of $n$, we have repeated our simulations for $t=200$ rounds and reported the maximum number of iterations over those $t=200$ instances. As it can be seen, the worst-case running time of Algorithm \ref{alg:iterative} is very small and does not exceed $\tau=22$ over all the instances. This shows that Algorithm \ref{alg:iterative} indeed converges very fast to a first-order stationary solution for the CIP even on large size networks.

\begin{figure}[t]
\vspace{-1.75cm}
\begin{center}
\includegraphics[totalheight=.3\textheight,
width=.45\textwidth,viewport=270 0 820 550]{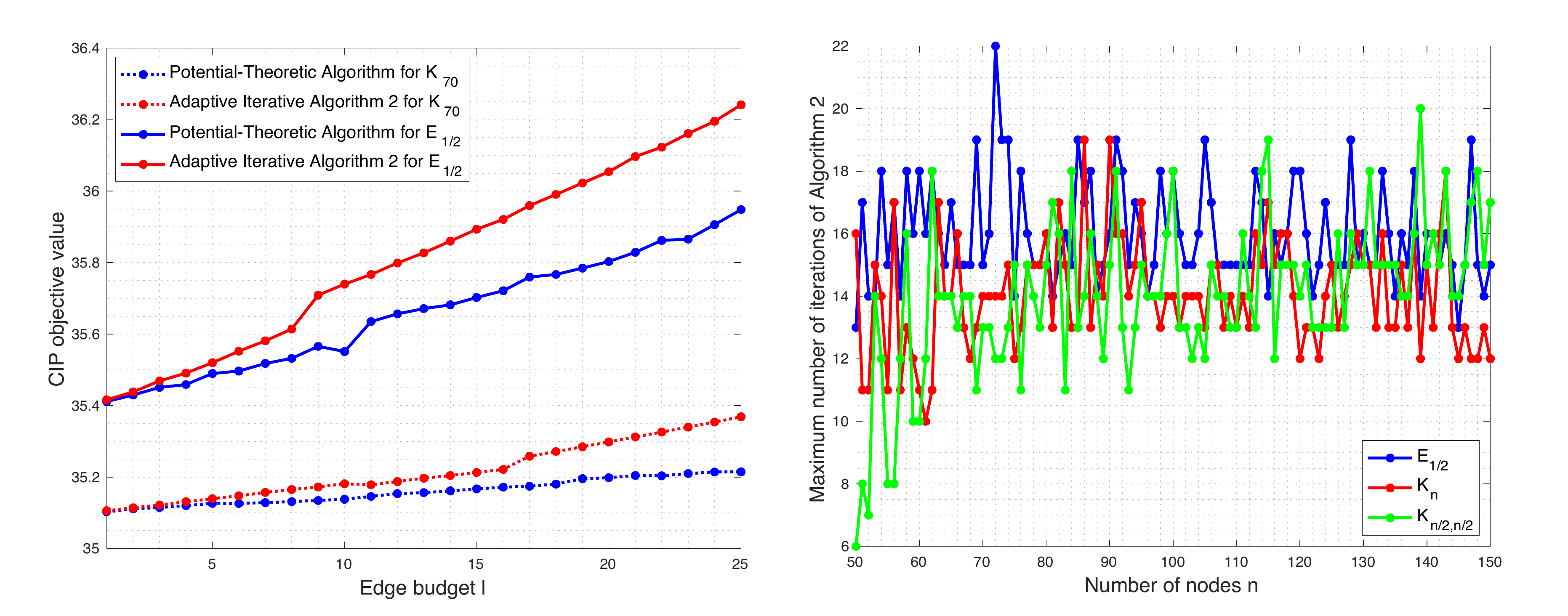} \hspace{0.4in}
\end{center}\vspace{-0.5cm}
\caption{The left figure shows the outperformance of Algorithm \ref{alg:iterative} over the potential-theoretic algorithm for different values of edge budget $\ell$. The right figure illustrates the worst-case iteration complexity of Algorithm \ref{alg:iterative} over different network structures.}\label{fig:iterate}
\end{figure}

\section{Conclusion}\label{sec:conclusion}
In this paper, we studied the consensus interdiction problem, where the goal is to maximize the convergence time of the consensus dynamics by removing a limited number of network edges.  By making a connection to the effective resistance interdiction problem, we showed that finding the optimal set of edges to interdict is strongly NP-hard, even on bipartite networks of diameter three with protected source/sink edges. We then showed that the effective resistance and the consensus interdiction problems are hard to approximate within an almost polynomial factor. Moreover, we devised a polynomial factor approximation algorithm for the effective resistance interdiction problem. Finally, we used a quadratic program to develop an iterative approximation algorithm for the consensus interdiction problem, and evaluate its good performance through numerical experiments.

\smallskip
\noindent
{\bf Acknowledgment:} I would like to thank Dr. Lap Chi Lau for pointing out a relevant reference on the problem. Also, I would like to thank anonymous reviewers for their helpful comments that allowed us to improve the paper. An earlier version of this paper has appeared in the 2021 American Control Conference (ACC) and is listed in \cite{etesami2021consensus}.

\bibliographystyle{IEEEtran}
\bibliography{thesisrefs}

\section*{Appendix I: Auxiliary Lemmas}\label{sec:appx1}

\smallskip
\begin{lemma}\label{lemm:paralel}
Let $\mathcal{G}''$ be a graph that is obtained from $\mathcal{G}'$ by replacing each edge in $E_R$ by $\frac{n^5}{\delta}$ parallel edges\footnote{If parallel edges are not allowed, one can replace each parallel edge by a path of length two.} of resistance $\frac{n^5}{\delta}$. Then, for any optimal cut $|\hat{E}|\leq m$, we must have $\hat{E}\subseteq E_{L}\cup E_1$. 
\end{lemma}
\begin{proof}
To derive a contradiction, let us assume that the optimal edge cut $\hat{E}$ removes $k\ge 1$ parallel edges from $\mathcal{G}''$. Then, by returning any one of those broken edges, we can reduce the $s-t$ effective resistance by at most $O(\frac{4\delta}{n^5})$. The reason is that, if a broken edge belongs to a branch $(v_i, t)$ containing at least $p\ge \frac{n^5}{\delta}-m\ge \frac{n^5}{2\delta}$ parallel edges of resistance $\frac{n^5}{\delta}$, then, returning that edge can reduce the $s-t$ effective resistance by at most 
\begin{align}\nonumber
\frac{\frac{n^5}{\delta}}{p}-\frac{\frac{n^5}{\delta}}{p+1}\leq \frac{n^5}{\delta p^2}\leq \frac{4\delta}{n^5}.
\end{align}

On the other hand, if we instead remove from $\mathcal{G}''$ one of the (at most $m$) edges $\{s, v_{ij}\}\in E_L$ with either $\{v_{ij},v_i\}\notin \hat{E}$ or $\{v_{ij},v_j\}\notin \hat{E}$, then the $s-t$ effective resistance increases by at least $\Omega(\frac{\delta}{m-1}-\frac{\delta}{m})=\Omega(\frac{\delta}{n^4})$. 
Thus, by returning a parallel edge in $\hat{E}$ and instead breaking an edge in $E_L$, the $s-t$ effective resistance strictly increases by $\Omega(\frac{\delta}{n^4})-\frac{4\delta}{n^5}>0$, which contradicts the optimality of the edge cut $\hat{E}$.
\end{proof} 

\smallskip
The following lemma establishes relation \eqref{eq:approx-effective-epsilon} given in the proof of Theorem \ref{thm:approximability}.

\begin{lemma}\label{lemm-hardness}
For any edge cut $\hat{E}$, one can find in polynomial time an edge cut of the form $\hat{E}_{S}$ given in \eqref{eq:E-S} such that $|\hat{E}_S|\leq \hat{E}$ and $\boldsymbol{\mbox{{\rm R}}}_{\rm eff}(\mathcal{G}'\setminus \hat{E}_S)\ge \boldsymbol{\mbox{{\rm R}}}_{\rm eff}(\mathcal{G}'\setminus \hat{E})-4\delta$.
\end{lemma}
\begin{proof}
Let us consider an arbitrary edge cut $\hat{E}$ in $\mathcal{G}'$, and partition the vertices $V_L$ into $\bar{N}$ and $N=V_{L}\setminus \bar{N}$. Here, $\bar{N}$ contains all the vertices in $V_L$ that are adjacent to the source node $s$ after interdiction (i.e., in the network $\mathcal{G}'\setminus \hat{E}$). Moreover, we let $\bar{S}\subseteq V_R$ contain all the nodes that have at least one neighbor in $\bar{N}$ after interdiction, and $S=V_R\setminus \bar{S}$.

Next, we modify $\hat{E}$ to $\hat{E}_S$ as follows. In the network $\mathcal{G}'$, let $M=\big\{v_{ij}\in N: v_j\in \bar{S}\big\}$ denote all the nodes in $N$ with at least one end point $v_j\in \bar{S}$. At each time we take one vertex $v_{ij}\in M$ and move it into $\bar{N}$ by processing it as follows. We first unbreak all the edges adjacent to $v_{ij}$ by setting 
$\hat{E}\leftarrow \hat{E}\setminus \{\{s,v_{ij}\},\{v_i,v_{ij}\},\{v_j,v_{ij}\}\}$. We then check whether $v_i\in S$, in which case we break the edge $\{v_{ij},v_i\}$ by setting $\hat{E}\leftarrow \hat{E}\cup \{v_{ij},v_i\}$. Note that the above process does not increase $|\hat{E}|$, as each time we unbreak at least one edge,  that is $\{s,v_{ij}\}$. After processing all the nodes in $M$ we obtain a new edge cut $\hat{E}$ with corresponding partitions $S$, $\bar{S}$, $N=N\setminus M$, $\bar{N}=\bar{N}\cup M$. In the final stage, we update $\hat{E}$ by unbreaking any edge that lies between $S$ and $N$ or between $\bar{S}$ and $\bar{N}$, i.e., we set $\hat{E}\leftarrow \hat{E}\setminus(E'[S,N]\cup E'[\bar{S},\bar{N}])$. It is easy to see that after the above modifications, the resulting edge cut equals $\hat{E}_{S}$, and moreover, $N$ contains all the edges in $\mathcal{G}'$ with both endpoints in $S$.

Any time that a vertex $v_{ij}\in M$ is processed, either an edge is broken (in which case the effective resistance increases) or at most two paths of length two between $v_j\in \bar{S}$ and $s$ are created. The key point is that for any $v_j\in \bar{S}$, there exists a small $2\delta$-resistance path $v_j\to v_{kj}\to s$ for some $v_{kj}\in \bar{N}$. Thus, any time that a new path between $s$ and some $v_j$ is created, we can ignore that path and instead upper-bound the decrease in the $s-t$ effective resistance by setting the resistance of the edges in the path $v_j\to v_{kj}\to s$ to $0$ (equivalently, shortcut $s$ and $v_j$). Thus, after processing all the nodes in $M$, the decrease in the $s-t$ effective resistance is at most the difference between the $s-t$ effective resistance of two identical networks (i.e., $\mathcal{G}'\setminus \hat{E}$ after removing all the nodes in $M$) where in one of them a subset of $\delta$-resistance edges are replaced by $0$-resistance edges. Using Thomson's principle, the difference between the $s-t$ effective resistance of those networks can be at most $2\delta$. 

Finally, after processing all the nodes in $M$ and moving them to $\bar{N}$, there is no edge between $N$ and $\bar{S}$ (or between $\bar{N}$ and $S$) in the interdicted network $\mathcal{G}'\setminus \hat{E}$. Therefore, addition of any edge between $S$ and $N$ during the final stage has no effect on the $s-t$ effective resistance. Moreover, for any edge $\{v_{ij},v_j\}, v_j\in \bar{S}, v_{ij}\in \bar{N}$ that may be added during the final stage, there existed a path $v_{j}\to v_{jk}\to s$ of resistance $2\delta$ between $v_j$ and $s$. Thus, using a similar argument as above, any time that such an edge $\{v_{ij},v_j\}$ is added, we can ignore it and instead shortcut the path $v_{j}\to v_{jk}\to s$. Using Thomson's principle, the total reduction in the $s-t$ effective resistance due to such shortcuttings can again be at most $2\delta$, which completes the proof. 
\end{proof}

\begin{lemma}\label{lemm:concave}
For any fixed vector $u\in \mathbb{R}^{n}$, the objective function $u'\big(\mathcal{L}(y)+\frac{J}{n}\big)u$ is a concave function of $y$, where $\mathcal{L}(y)=I-[P(y)]^2$.
\end{lemma}
\begin{proof}
Since $u'\big(\mathcal{L}(y)+\frac{J}{n}\big)u=u'\mathcal{L}(y)u+u'\frac{J}{n}u$, we only need to show that $u'\mathcal{L}(y)u$ is a concave function of $y$. Using \eqref{eq:p-y},  for any $y^1,y^2\in \mathbb{R}^{m}$ and $\lambda\in [0,1]$, we have
\begin{align}\nonumber
P(\lambda y^1+(1-\lambda)y^2)=\lambda P(y^1)+(1-\lambda)P(y^2). 
\end{align}
For any fixed vector $u$, we have
\begin{align}\nonumber
&u'\mathcal{L}(\lambda y^1+(1-\lambda)y^2)u=u'u-u'\big[\lambda P(y^1)+(1-\lambda)P(y^2)\big]^2u,\cr 
&\lambda u'\mathcal{L}(y^1)u+(1-\lambda) u'\mathcal{L}(y^2)u=u'u-\lambda u' [P(y^1)]^2 u-(1-\lambda)u'[P(y^2)]^2 u.
\end{align}
Therefore, to show that $u'\mathcal{L}(y)u$ is a concave function of $y$, it is enough to show that 
\begin{align}\nonumber
u'\big[\lambda P(y^1)+(1-\lambda)P(y^2)\big]^2u\leq u'\big(\lambda [P(y^1)]^2+(1-\lambda)[P(y^2)]^2\big)u, \ \forall u.
\end{align}
The above inequality is also true because
\begin{align}\nonumber
\lambda [P(y^1)]^2+(1-\lambda)[P(y^2)]^2-\big[\lambda P(y^1)+(1-\lambda)P(y^2)\big]^2=\lambda(1-\lambda)\big[P(y^1)-P(y^2)\big]^2,
\end{align}
which is a positive semi-definite matrix.
\end{proof} 

\smallskip
\begin{lemma}\label{lemm:gradient-expansion}
Let $f(u,y)=u'\big(\mathcal{L}(y)+\frac{J}{n}\big)u$, where $\mathcal{L}(y)=I-[P(y)]^2$ and $P(y)$ is given by \eqref{eq:p-y}. Then, for any link $\{i,j\}\in E, i\neq j$, we have 
\begin{align}\label{eq:gradient-formula-lemma}
\nabla_{y_{ij}} f(u,y)&=2p_{ij}(u_i-u_j)^2\big([P(y)]_{ii}+[P(y)]_{jj}-2[P(y)]_{ij}\big)\cr 
&\qquad+\!2p_{ij}\sum_{k\neq i,j}\big((u_k-u_j)^2-(u_k-u_i)^2\big)\big([P(y)]_{ik}-[P(y)]_{jk}\big).
\end{align}
\end{lemma}
\begin{proof}
As $P(y)$ is a symmetric matrix, using direct calculation for any $i'\neq j'$,
\begin{align}\label{eq:P2-expansion}
[P(y)]^2_{i'j'}=\sum_{k}[P(y)]_{i'k}[P(y)]_{j'k}&=2p_{i'j'}y_{i'j'}+\sum_{k\neq i',j'}p_{i'k}p_{j'k}y_{i'k}y_{j'k}\cr 
&-\sum_{k\neq i'}p_{i'k}p_{i'j'}y_{i'k}y_{i'j'}-\sum_{k\neq j'}p_{j'k}p_{i'j'}y_{j'k}y_{i'j'}.
\end{align}
Since $\mathcal{L}(y)=I-[P(y)]^2$ is a Laplacian matrix, $u'\mathcal{L}(y)u=\sum_{i'\neq j'}[P(y)]^2_{i'j'}(u_{i'}-u_{j'})^2$. Using the fact that $\nabla_{y} f(u,y)=\nabla_{y}u'\mathcal{L}(y)u$, for any $i\neq j$, we can write
\begin{align}\label{eq:single-gradient-expression}
\nabla_{y_{ij}}f(u,y)=\sum_{i'\neq j'}\big(\nabla_{y_{ij}}[P(y)]^2_{i'j'}\big)(u_{i'}-u_{j'})^2.
\end{align}
Using \eqref{eq:P2-expansion}, one can compute $\nabla_{y_{ij}}[P(y)]^2_{i'j'}$ for all combinations of $i'\neq j', i\neq j$ as 
\begin{align}\nonumber
\nabla_{y_{ij}}[P(y)]^2_{i'j'}=\begin{cases}
p_{ij}\big(2-4p_{ij}y_{ij}\!-\!\sum_{k\neq i,j}p_{ik}y_{ik}\!-\!\sum_{k\neq i,j}p_{jk}y_{jk}\big) & \mbox{if}\ i'=i,\ j'=j\\
p_{ij}\big(2-4p_{ij}y_{ij}\!-\!\sum_{k\neq i,j}p_{ik}y_{ik}\!-\!\sum_{k\neq i,j}p_{jk}y_{jk}\big) & \mbox{if}\ i'=j,\ j'=i\\
p_{ij}\big(p_{ii'}y_{ii'}-p_{ji'}y_{ji'}\big) & \mbox{if}\ i'\neq i,\ j'=j\\
p_{ij}\big(p_{ji'}y_{ji'}-p_{ii'}y_{ii'}\big) & \mbox{if}\ i'\neq j,\ j'=i\\
p_{ij}\big(p_{jj'}y_{jj'}-p_{ij'}y_{ij'}\big) & \mbox{if}\ i'= i,\ j'\neq j\\
p_{ij}\big(p_{ij'}y_{ij'}-p_{jj'}y_{jj'}\big) & \mbox{if}\ i'= j,\ j'\neq i\\
0 & \mbox{if}\ i'\neq i,j,\ j'\neq i,j.
\end{cases} 
\end{align}
If we substitute the above expressions into \eqref{eq:single-gradient-expression} and sum over all $i'\neq j'$, we obtain
\begin{align}\nonumber
\nabla_{y_{ij}}f(u,y)&=2p_{ij}(u_i-u_j)^2\big(2-4p_{ij}y_{ij}\!-\!\sum_{k\neq i,j}p_{ik}y_{ik}\!-\!\sum_{k\neq i,j}p_{jk}y_{jk}\big)\cr 
&+2p_{ij}\!\!\!\sum_{k\neq i,j}(u_k-u_j)^2\big(p_{ik}y_{ik}-p_{jk}y_{jk}\big)\!+\!2p_{ij}\!\!\!\sum_{k\neq i,j}(u_k-u_i)^2\big(p_{jk}y_{jk}-p_{ik}y_{ik}\big)\cr 
&=2p_{ij}(u_i-u_j)^2\big(2-2p_{ij}y_{ij}\!-\!\sum_{k\neq i}p_{ik}y_{ik}\!-\!\sum_{k\neq j}p_{jk}y_{jk}\big)\cr 
&+2p_{ij}\sum_{k\neq i,j}\big((u_k-u_j)^2-(u_k-u_i)^2\big)\big(p_{ik}y_{ik}-p_{jk}y_{jk}\big)\cr
&=2p_{ij}(u_i-u_j)^2\big([P(y)]_{ii}+[P(y)]_{jj}-2[P(y)]_{ij}\big)\cr 
&+2p_{ij}\sum_{k\neq i,j}\big((u_k-u_j)^2-(u_k-u_i)^2\big)\big([P(y)]_{ik}-[P(y)]_{jk}\big),
\end{align} 
where the last equality uses $[P(y)]_{ij}=p_{ij}y_{ij}$, $[P(y)]_{ii}=1-\sum_{k\neq i}p_{ik}y_{ik}$, and $[P(y)]_{jj}=1-\sum_{k\neq j}p_{jk}y_{jk}$. 
\end{proof}

\section*{Appendix II: A Correction to the Past Literature}\label{sex:apx2}

Here, we first provide a counterexample to show that the potential-theoretic algorithm given for the adversary in a series of works \cite{khanafer2015robust,khanafer2013robust,khanafer2012consensus} is not optimal, and then point out the source of error in their proofs. Here, we note that the model in \cite{khanafer2015robust} is written for the continuous-time with the objective function $\int_{t=0}^{T} k(t) \|x(t) - \bar{x}\|^2$ and finite horizon $T$. For simplicity of presentation, we provide a counterexample for the discrete-time model with the objective function $\sum_{t=0}^{T} k(t) \|x(t) - \bar{x}\|^2$ and infinite horizon $T=\infty$. However, one can always choose $T$ to be a sufficiently large finite number and partition $[0, T]$ into equi-length discrete intervals of sufficiently small length. Since our analysis is robust to any $o(1)$-perturbations in the objective value (e.g., due to time truncation/discretization), our counterexample remains valid even for the continuous-time model.

Let us consider a special case of the problem proposed in \cite{khanafer2015robust}, where there is no network designer or, equivalently, there is a network designer with zero budget $b=0$.\footnote{Even for small $b>0$, the robustness of the proposed counterexample still invalidates the global optimality of the potential-theoretic strategy.} We set the time horizon to $T=\infty$, and the kernel function to $k(t)=2, \forall t\in [0, T)$. Moreover, we choose the budget of the adversary to be $\ell=1$ and set the dwell time in \cite[Assumption 1]{khanafer2015robust}, which is the minimum time between consecutive switching times of the adversary's strategy, to be $\tau=T$. In other words, we consider a simpler static version of the problem in \cite{khanafer2015robust}, wherein the adversary can interdict the network only once at time $t=0$. We define the averaging matrix $A$ in \cite{khanafer2015robust} to be the negative Laplacian $A:=-(I-P)$, where $P$ is the conductance matrix of a cycle with three nodes $\{1,2,3\}$ (i.e., a triangle) given by
\begin{align}\nonumber
P=\begin{pmatrix} \vspace{0.1cm}
\frac{17}{30} & \frac{1}{3} & \frac{1}{10} \\ \vspace{0.1cm}
\frac{1}{3} & \frac{1}{3} & \frac{1}{3}\\ 
\frac{1}{10} & \frac{1}{3} & \frac{17}{30} \\
\end{pmatrix}.
\end{align}
Finally, we choose the initial vector to be $x_0=\boldsymbol{{\rm e}}_1-\boldsymbol{{\rm e}}_3=(1,0,-1)'$. Based on the above setting, the potential-theoretic strategy given in \cite[Theorem 1]{khanafer2015robust} computes the power dissipation associated to each of the three edges in the cycle and breaks the one of highest power dissipation \cite[Remark 3]{khanafer2015robust}. Since the power dissipation associated to the edges $\{1,2\}, \{2,3\}$, and $\{1,3\}$ are $\frac{1}{3}(1-0)^2=\frac{1}{3}$, $\frac{1}{3}(0+1)^2=\frac{1}{3}$, and $\frac{1}{10}(1+1)^2=\frac{4}{10}$, respectively, the potential-theoretic strategy must break $\{1,3\}$.  

Next, we show that the optimal strategy for the adversary is to break the link $\{2,3\}$. (Note that by symmetry the case of breaking $\{1,2\}$ is the same as breaking $\{2,3\}$.) Using Lemma \ref{lemm:eff-cast}, the optimal strategy for the adversary is to break a link $\{i,j\}$ that maximizes $\mbox{R}_{\rm eff}\big((P\!\setminus\!\{i,j\})^2\big)$, where $\mbox{R}_{\rm eff}$ denotes the effective resistance between nodes $s=1$ and $t=3$. By considering either of the cases, we have
\begin{align}\nonumber
P\!\setminus\!\{2,3\}\!=\!\begin{pmatrix} \vspace{0.1cm}
\frac{17}{30} & \frac{1}{3} & \frac{1}{10} \\ \vspace{0.1cm}
\frac{1}{3} & \frac{2}{3} & 0\\ 
\frac{1}{10} & 0 & \frac{9}{10} \\
\end{pmatrix}, \  \
P\!\setminus\!\{1,3\}\!=\!\begin{pmatrix} \vspace{0.1cm}
\frac{2}{3} & \frac{1}{3} & 0 \\ \vspace{0.1cm}
\frac{1}{3} & \frac{1}{3} & \frac{1}{3}\\ 
0 & \frac{1}{3} & \frac{2}{3} \\
\end{pmatrix}.
\end{align}
By computing the square of the above matrices, we get
\begin{align}\nonumber
\big(\!P\!\setminus\!\{2,3\}\!\big)^2\!\!\!=\!\!\begin{pmatrix} \vspace{0.1cm}
\frac{199}{450} & \frac{37}{90} & \frac{11}{75} \\ \vspace{0.1cm}
\frac{37}{90} & \frac{5}{9} & \frac{1}{30}\\ 
\frac{11}{75} & \frac{1}{30} & \frac{41}{50} \\
\end{pmatrix}, \ \ 
\big(\!P\!\setminus\!\{1,3\}\!\big)^2\!\!\!=\!\!\begin{pmatrix} \vspace{0.1cm}
\frac{5}{9} & \frac{1}{3} & \frac{1}{9} \\ \vspace{0.1cm}
\frac{1}{3} & \frac{1}{3} & \frac{1}{3}\\ 
\frac{1}{9} & \frac{1}{3} & \frac{5}{9} \\
\end{pmatrix}.
\end{align}
Finally, an easy calculation reveals that
\begin{align}\nonumber
\mbox{R}_{\rm eff}\big((P\setminus\{2,3\})^2\big)=\frac{400}{71}>\frac{18}{5}=\mbox{R}_{\rm eff}\big((P\setminus\{1,3\})^2\big),
\end{align}
which shows that the adversary's optimal strategy is to break the link $\{2,3\}$. This completes the counterexample.  

The error in the work \cite{khanafer2015robust} is because of incorrect use of Taylor approximation to conclude global optimality from a local property of the consensus dynamics. In fact, in our counterexample, we purposefully chose the horizon's length sufficiently large so that the Taylor approximation error in \cite{khanafer2015robust} manifests itself by moving from local to global analysis.  More precisely, the Taylor approximation in \cite[Equation (17)]{khanafer2015robust} is only valid for a small length interval $[t_0, t_0+2\delta]$, and the analysis cannot be repeated and generalized to arbitrary length intervals. In particular, the mimicking behavior of the proposed policies in \cite[Theorem 1]{khanafer2015robust} is not sufficient to guarantee the closeness of the long-run trajectories. The same issue also exists in other works \cite{khanafer2013robust,khanafer2012consensus}. In fact, as we showed in Theorem \ref{thm:scaling-resistance}, such an error is fundamental and cannot be fixed unless P = NP. It is worth noting that although the potential-theoretic algorithm is not globally optimal, as we justified in our simulations, it still constitutes a good suboptimal algorithm for the CIP. 

\end{document}